\documentclass[acmsmall,screen,authorversion,nonacm]{acmart}
\AtBeginDocument{%
  }

\setcopyright{acmlicensed}
\copyrightyear{2026}
\acmYear{2026}
\acmDOI{XXXXXXX.XXXXXXX}

\makeatletter
\let\Bbbk\@undefined
\makeatother
\usepackage{my-paper-preamble}
\usepackage{paracol}
\usepackage{enumitem}
\usepackage{tikz-cd}
\allowdisplaybreaks

\newlist{enumeratealigned}{enumerate}{1}
\setlist[enumeratealigned]{
  label=(\arabic*).,
  leftmargin=*,
  topsep=0.4\baselineskip,
  before=\leavevmode  
}

\newcommand{\FV}{\operatorname{FV}}

\newcommand{\Ch}{\operatorname{Ch}}
\newcommand{\Var}{{\textit{Var}}}
\newcommand{\Chan}{{\textit{Chan}}}
\newcommand{\env}{{\text{env}}}
\usepackage{xcolor}
\usepackage{tikz-cd}
\usepackage{etoolbox}
\usepackage{mathtools}

\newcommand{\interact}{\mathcal{I}}

\newcommand{\CX}{{\text{CX}}}
\newcommand{\CH}{\operatorname{CH}}
\newcommand{\fin}{{\text{fin}}}

\newcommand{\st}[2]{\langle #1, #2 \rangle}

\newcommand{\sstep}[2][]{\mathrel{{\overset{#2}{\mapsto}}\vphantom{\mapsto}^{#1}}}
\newcommand{\Sstep}[1]{\mathrel{{\overset{#1}{\Mapsto}}}}

\newcommand{\parcomp}{\parallel}

\newcommand{\Obs}{\mathcal{O}}
\newcommand{\Obsvs}{\mathbb{O}}
\newcommand{\Domain}{\Delta}

\newcommand{\supp}{\operatorname{supp}}
\newcommand{\density}{\operatorname{density}}

\newcommand{\Proc}{\textit{Proc}}
\newcommand{\State}{\textit{State}}
\newcommand{\Trace}{\textit{Trace}}

\newcommand{\recons}{\mathcal{R}}
\newcommand{\partition}{\pi}

\newcommand{\obseq}{\approx_{\text{obs}}}
\newcommand{\obssim}[1]{\sim^{#1}_{\text{obs}}}

\newcommand{\sink}{\operatorname{sink}}
\newcommand{\subsim}{\sim}
\newcommand{\subeq}{\approx}

\newcommand{\Bell}{\operatorname{Bell}}

\newcommand{\concat}{\mathbin{+\mkern-5mu+}}

\usetikzlibrary{calc}

\newif{\ifproofdetails}
\proofdetailsfalse

\begin{document}
\title{Distributed Semantics for Distributed Quantum Computing}

\author{Jun Inoue}
\email{jun.inoue@aist.go.jp}
\orcid{0000-0002-2939-8337}
\affiliation{%
  \institution{National Institute of Advanced Industrial Science and Technology}
  \city{Ikeda}
  \state{Osaka}
  \country{Japan}
}

\begin{abstract}
  We present a quantum process calculus that can split the system 
  state along process boundaries and follow the evolution of each 
  process in isolation, without losing information about the joint 
  state---a property we call spatial compositionality.
  Compositionality is the key to reasoning about any complex system, 
  yet quantum process calculi have struggled to provide its spatial 
  kind, which would enable analyzing a system one process at a time.
  Many a quantum process calculi have been proposed, but they 
  invariably rely on a global state representation based on state 
  vectors or density matrices, with no known way to split them 
  without losing information about entanglement.

  We propose to model quantum states with Deutsch-Hayden descriptors 
  instead, which provide a modular representation of qubit states and 
  their evolution.
  We adapt these descriptors to allow arbitrary splitting and merging 
  of the store of qubits, leading to an unusual process calculus in 
  which qubit transfer messages carry the actual state of the qubit, 
  where existing calculi transfer only a reference.
  The calculus gives localized views of system state visible to each 
  process, which can be assembled back together into the joint state.
  We define a notion of process equivalence with extensive 
  justification grounded in physics and show a bisimulation whose 
  soundness proof is simplified by spatial compositionality.
  The calculus can model open systems entangled with external 
  processes, and we demonstrate this capability on a fragment of the 
  BB84 key distribution protocol.
  This exercise shows that Deutsch-Hayden descriptors can 
  successfully track qubit movements across process and system 
  boundaries, though it needs help from density matrices to reason 
  about information flow.
\end{abstract}



\received{9 July 2026}

\maketitle

\section{Introduction}
\label{sec:intro}

In this article, we present a quantum process calculus whose 
semantics can follow the evolution of each process in isolation, in 
such a way that the results can be uniquely assembled back into the 
evolved joint system state, regardless of entanglement.
We call this property \emph{spatial compositionality}.

An important application of quantum computation is in communication 
protocols, such as cryptographic key distribution.
These protocols are distributed quantum algorithms, which have a 
natural home in quantum process calculi.
A number of process calculi have been proposed against this backdrop, 
such as QPAlg \cite{LalireJorrand2004ProcessAlgebraicApproach}, CQP 
\cite{GayNagarajan2005CommunicatingQuantumProcesses}, qCCS 
\cite{FengEtAl2007ProbabilisticBisimulationsQuantum}, and lqCCS 
\cite{CeragioliEtAl2024QuantumBisimilarityBarbs}.

All of these calculi lump the quantum states held by all processes 
into a single, monolithic state.
Formally, they define a labeled transition system whose nodes are 
configurations that look like $\st{P}{\rho}$ (plus bookkeeping data), 
where $\rho$ is a state vector or density matrix spanning not just 
the qubits held by process $P$ but all other processes that may be 
running in parallel.
But those representations are not modular, giving no lossless way to 
split the system state into process-local views.
Partial trace, the only generally accepted tool for narrowing the 
view of a density matrix, produces local views that collectively 
cannot recover the global state---it loses information.

Writ large, a monolithic semantics cannot model fully open system 
that interact and entangle with the outside world.
For example, qCCS has a construct for external input, but it requires 
that the incoming qubit is unentangled with the local ones.
Writ small, a monolithic semantics cannot analyze the system's 
processes individually, because the effects of operation in one 
process ripples through the global quantum state.

Monolithic semantics makes it awkward to justify a style of reasoning 
that comes most naturally in classical distributed systems: to break 
them apart into local processes, to reason about them and their 
interactions locally, and infer the behavior of the whole.
Modularity is the cornerstone of reasoning about complex systems, but 
to our knowledge, this feature has so far been denied in quantum 
computing.
A spatially compositional semantics lays much needed groundwork for 
modular reasoning by prescribing the rules for dealing with fragments 
of entangled systems.

The idea is to stop modeling quantum state by density matrices and 
instead use Deutsch-Hayden (DH) descriptors 
\cite{DeutschHayden2000InformationFlowEntangled, 
  Bedard2021ABCDeutschHayden}, which do away with the global state 
vector and tracks modifications locally to each qubit.
That is, an operation on qubit $x$ does not affect the DH descriptor 
of qubit $y$, even when $x$ and $y$ are entangled.
Crucially, unlike partial traces, the individual descriptors for $x$ 
and $y$ suffice to reconstruct the state of the joint $x$-$y$ system.
We adapt this formalism to describe a quantum store that can be split 
apart, regardless of entanglement status.
This state representation naturally factors the semantics of 
distributed quantum systems into process-local views, allowing us to 
reason about individual processes in isolation.

This paper's contributions are as follows.
\begin{itemize}
\item After reviewing DH descriptors (\cref{sec:dh}), we adapt them 
  to allow arbitrary splitting and merging along spatial boundaries 
  (\cref{sec:splittable-dh}).
  The key modification is to move from indexing qubits to naming 
  them, providing stable identifiers that withstand splitting and 
  merging.
\item We present DH-CCS, a quantum variant of CCS representing 
  quantum states by DH descriptors (\cref{sec:dh-ccs}).
  Its unique feature is that a message $c!q$ carries not just a 
  qubit's name but its whole state, constructed as a single-qubit 
  fragment of the sending process' store.
\item We explain how DH-CCS unifies allocation, discarding, 
  measurement, and communication into a single common action: 
  transfer of qubits across the system boundary 
  (\cref{sec:unification}).
  Unfortunately, classical control by measurement results fall 
  outside of this unification.
  DH-CCS leaves out classical control, though we explain how this is 
  not a grave limitation.
\item We prove that DH-CCS is spatially compositional 
  (\cref{sec:metatheory}), namely that the result of independently 
  evolving processes $P$ and $Q$ can predict the result of evolving 
  them together in a parallel composition $P\parcomp Q$.
  We follow up with an analysis of why density matrices do not seem 
  to support such reasoning, showing what makes DH descriptors 
  different---namely, descriptors encode bounded information about 
  evolution history.
\item We define an observer equivalence, accompanied with a careful 
  derivation of requirements vetted against what is allowable by the 
  laws of physics (\cref{sec:bisimulation}).
  Comparing current measurement probabilities is insufficient because 
  DH descriptors capture the structure of entanglement with the 
  outside world.
  Differences in that structure only manifest through interactions 
  over time, so continuous monitoring is needed to tell them apart.
  We also show a bisimilarity whose soundness proof is simplified by 
  spatial compositionality.
\item We illustrate how DH-CCS can be used to model and reason about 
  an open system, taking a fragment of the BB84 quantum key 
  distribution protocol 
  \cite{BennettBrassard2014QuantumCryptographyPublic} as an example 
  (\cref{sec:example}).
  This example showcases what DH descriptors are and are not good 
  for, given current understanding of them.
  Briefly: they are good for monitoring the dynamics of qubit 
  movements across process and system boundaries but not for accurate 
  accounting of information flow.
\end{itemize}

\section{Background: Deutsch-Haynes Descriptors}
\label{sec:dh}

In this section, we review DH descriptors.
What follows is mostly a condensed summary of the excellent 
introduction by B\'{e}dard \cite{Bedard2021ABCDeutschHayden}, with 
some materials from Horsman and Vedral 
\cite{HorsmanVedral2007DevelopingDeutschHayden}.

A conventional description of quantum computing tracks the evolution 
of state vectors, usually initialized to $\ket{0^n}$.
One asks how the state vector changes under gate 
applications---represented by a unitary $U$---and then ask what a 
measurement of that evolved state produces.
A popular way to represent measurements is the observable, an 
Hermitian operator $\Obs$ which, when sandwiched like 
$\bra{0^n}U^\dagger\Obs U\ket{0^n}$, produces an expectation value 
for the evolved state $U\ket{0^n}$.

DH descriptors, by contrast, evolves the observables, tracking ``the 
observable that represents $U$ followed by measurement by $\Obs$'' as 
a description of the system state.
The evolved observable is $U^\dagger \Obs U$, which when sandwiched 
by the initial state vector $\ket{0^n}$, produces the same 
$\bra{0^n}U^\dagger \Obs U\ket{0^n}$.

There are infinitely many observables, but it suffices to track the 
evolution of just two observables per qubit, the Pauli $X$ and $Z$.
This is because the Pauli group 
$\bigl\{\bigotimes_{j=1}^n P_j \mid \forall j.\ P_j \in 
\{1,X,Y,Z\}\bigr\}$ is a basis for $\mathbb{C}^{2^n\times 2^n}$, so 
any matrix can be expressed as a linear combination over this group, 
and the $Y$'s can be eliminated by $Y = iXZ$.
Thus, there is an expansion of any $\Obs$ in a multilinear 
combination of $X_j := 1^{j-1}\otimes X\otimes 1^{n-j}$ and 
$Z_j := 1^{j-1}\otimes Z\otimes 1^{n-j}$, respectively Pauli $X$ and 
$Z$ targeting the $j$-th qubit.
We may capture this expansion in a function $f_\Obs$ as
\begin{equation}\label{eq:fU-def}
  \Obs = f_\Obs ((X_1,Z_1), \ldots, (X_n, Z_n)).
\end{equation}
This is called the \emph{functional expansion} of $\Obs$.
Unitary conjugation distributes over ring operations, so
\begin{equation}\label{eq:fU-conjugate}
  U^\dagger \Obs U = f_\Obs ((U^\dagger X_1U, U^\dagger Z_1U), \ldots, 
  (U^\dagger X_nU, U^\dagger Z_nU))
\end{equation}
gives a formula to recover the evolved form of any $\Obs$ from the 
evolved forms of the $X_j$'s and $Z_j$'s.

A \emph{DH descriptor set} for an $n$-qubit system is an array of 
observables giving the evolved forms of $X_j$'s and $Z_j$'s.
Specifically, it is a mapping 
$q : \{1,\ldots,n\} \rightarrow \Obsvs_n^2$ where $\Obsvs_n$ is the 
set of observables (i.e.\ Hermitian matrices) on the space of $n$ 
qubits $\mathbb{C}^{2^n}$.
Each descriptor is initialized to $q_0(j):= (X_j, Z_j)$, and evolving 
it by $U$ turns it into $q = U^\dagger q_0 U$ where matrix 
multiplication is understood to distribute over the array of matrices 
$q_0$.
The pair that makes up a single-qubit descriptor will be indexed by 
formal symbols $X$ and $Z$, so for example $q_0(j)_X = X_j$ and 
$q_0(j)_Z = Z_j$ while in general $q(j) = (q(j)_X, q(j)_Z)$.
A DH descriptor set can be seen as a set of arguments to the $f_\Obs$ 
function, so \cref{eq:fU-def} can be written more concisely as 
$\Obs = f_\Obs (q_0)$, and \cref{eq:fU-conjugate} as 
$U^\dagger \Obs U = f_\Obs (U^\dagger q_0 U) = f_\Obs (q)$.

The direct evolution formula $q = U^\dagger q_0 U$ is problematic for 
incremental applications.
Applying unitary gates $U_1, U_2, \ldots, U_k$ in that order gives 
$U_1^\dagger U_2^\dagger \cdots U_k^\dagger q_0U_k \cdots U_2 U_1$, 
which can only be directly calculated by applying $U_k$ first and 
$U_1$ last.
The solution is to reuse $f_U$, which is in fact defined for all 
matrices $U$ (not just observables or unitaries) because the Pauli 
group spans all of $\mathbb{C}^{2^n\times 2^n}$.
When a modified update formula $q \mapsto f_U(q)^\dagger q\, f_U(q)$ 
is iterated for $U_1, \ldots, U_k$ in that desired order:
\begin{align}
  q_1 &= f_{U_1}(q_0)^\dagger q_0\, f_{U_1}(q_0) \notag \\
  q_2 &= f_{U_1}(q_1)^\dagger f_{U_2}(q_1)^\dagger q_0\, f_{U_1}(q_0)f_{U_2}(q_1) \notag \\
  q_3 &= f_{U_3}(q_3)^\dagger f_{U_2}(q_2)^\dagger f_{U_1}(q_1)^\dagger
  q_0\,f_{U_1}(q_1) f_{U_2}(q_2)f_3(q_3) \text{ etc.}, \notag
\end{align}
the result works out \cite{Bedard2021ABCDeutschHayden} to be
\begin{equation}\label{eq:incremental-update}
  q_k = U_1^\dagger U_2^\dagger \cdots U_k q_0 U_k \cdots U_2 U_1.
\end{equation}

It's important to note that DH descriptor evolution \emph{must} start 
with $q_0$.
This ensures a critical invariant necessary for the modified update 
formula to work out as \eqref{eq:incremental-update}.

\begin{definition}\label{defn:su2-algebra}
  The \emph{Pauli equations} on a partial DH descriptor set 
  $q : \{1..n\} \rightarrow \Obsvs_N^2$ ($n \leq N$) consist of the 
  following equations for all $j$ and $k \neq j$ in $\dom q$:
  \begin{itemize}
  \item 
    $q(j)_X^\dagger = q(j)_X \land 
    q(j)_Z^\dagger = q(j)_Z$ \hfill (descriptors are Hermitian)
  \item $q(j)_X^2 = q(j)_Z^2 = 1$ \hfill (descriptors are 
    involutions)
  \item $q(j)_X q(j)_Z = - q(j)_Z q(j)_X$ 
    \hfill (distinct same-qubit descriptors anti-commute)
  \item 
    $\forall \ell,m \in \{X,Z\}.\ q(j)_\ell q(k)_m 
    = q(k)_m q(j)_\ell$ \hfill (different-qubit descriptors 
    commute).
  \end{itemize}
\end{definition}

These equations characterize the valid DH descriptor sets.
The initial set $q_0$ satisfies them, and unitary update 
$q \mapsto U^\dagger q U$ preserves them.
We'll see later that the converse also holds: any set satisfying the 
equations has the form $U^\dagger q_0 U$, letting us complete partial 
frames with missing qubits.

The most salient fact about DH descriptors is that they are 
\emph{modular}: i.e.\ applying a unitary affects only the descriptors 
of the qubits that the unitary directly touches.
Other qubits' descriptors are unmodified, whether or not they are 
entangled with the target\footnote{Some unitary gates like $\CX$ 
  distinguish between control qubits and target qubits, but in this 
  article, we use ``target'' as a shorthand for all qubits that are 
  nontrivially involved with the unitary.}
qubits.

\begin{fact}[modularity]\label{thm:dh-modularity}
  Suppose $q$ is a set of DH descriptors of the form 
  $q = U^\dagger q_0 U$ with $U$ a unitary.
  If $V$ is another unitary that targets only some of the qubits, and 
  $q'$ is the result of evolving $q$ by $V$, then $q$ and $q'$ agree 
  on the non-target qubits.
  Formally, (assuming without loss of generality that $V$ targets the 
  first $k$ out of $n$ qubits) if $V = W\otimes 1^{n-k}$ for some 
  unitary $W$ then $q'(j) = q(j)$ for $j > k$.
  \begin{proof}
    Because $W$ is an operator in the subspace of the first $k$ 
    qubits, it can be written as a multilinear combination of $X_j$'s 
    and $Z_j$'s for $j \leq k$.
    This means $f_W(q)$ picks up only those $q(j)$ with $j \leq k$.
    By the Pauli equations, they commute with $q(j)$ for $j > k$, so 
    noting $f_V(q) = f_W(q) \otimes 1^{n-k}$, we have 
    $q'(j) = f_V(q)^\dagger q(j)\, f_V(q) = f_V(q)^\dagger f_V(q) 
    q(j) = q(j)$ for $j > k$.
  \end{proof}
\end{fact}

\begin{example}
  Let us consider the creation of a Bell state 
  $\frac{1}{\sqrt{2}}(\ket{00} + \ket{11})$, followed by Hamadard on 
  qubit 1 to see its effects.
  The initial DH descriptors, encoding a $\ket{00}$ state, are
  \begin{mathpar}
    q_0 = [q_0(1), q_0(2)] = [(X_1, Z_1), (X_2, Z_2)].
  \end{mathpar}
  In preparation of the Bell state, we apply Hadamard to qubit 1.
  Note $H_1 = (X_1+Z_1)/\sqrt{2}$, hence 
  $f_{H_1}(q) = (q(1)_1 + q(1)_2)/\sqrt{2}$.
  Then, calculation gives the encoding of $\ket{+0}$ as
  \begin{mathpar}
    q_1 = f_{H_1}(q_0)^\dagger q_0\,f_{H_1}(q_0) = [(Z_1, X_1), 
    (X_2, Z_2)].
  \end{mathpar}
  Then applying 
  $\CX_{1\rightarrow 2} = \frac{1}{2}(1\otimes1 + Z_1 + X_2 - Z_1 
  X_2)$ gives the encoding of 
  $\frac{1}{\sqrt{2}}(\ket{00} + \ket{11})$ as:
  \begin{mathpar}
    q_2 = f_{\CX}(q_1)^\dagger q_1 f_{\CX}(q_1) = [(Z_1 X_2, X_1), 
    (X_2, X_1 Z_2)].
  \end{mathpar}
  Notice how $X_2$, originally a descriptor for qubit 2, now lives in 
  the descriptor for qubit 1.
  This is how DH descriptors represent entanglement: Paulis of 
  entangled qubits show up in the local descriptor.
  Let us now apply another Hadamard on qubit 1.
  The result works out to be:
  \begin{mathpar}
    q_3 = f_{H_1}(q_2)^\dagger q_2 f_{H_1}(q_2) = [(X_1, Z_1 X_2), 
    (X_2, X_1 Z_2)].
  \end{mathpar}
  The descriptor of qubit 2 has not changed at all, despite 
  entanglement between the qubits.
\end{example}

Let us show how to extract the corresponding density matrix from a 
set of DH descriptors, including a handling of discarding.
B\'{e}dard \cite{Bedard2021ABCDeutschHayden} does not work out this 
extraction, but related discussions are offered by Horsman and Vedral 
\cite{HorsmanVedral2007DevelopingDeutschHayden}.

Recall that the density matrix of a state $\ket{\psi}$ (pure or 
mixed) is the expectation-weighted sum 
$\sum_{\Obs} \langle \psi|\Obs|\psi\rangle \Obs$ with $\Obs$ running 
over any ONB of observables with respect to the Hilbert-Schmidt inner 
product.
The Pauli group is one such ONB.
Letting $S := \{1,X,Y,Z\}$, each string $s \in S^n$ specifies a Pauli 
group element $P_s := \prod_{j=0}^{n-1} (s_j)_j$ where $(s_j)_j$ means 
the matrix indicated by $s_j$ (i.e.\ the $j$-th symbol in the string 
$s$) acting on qubit $j$; for example, 
$P_{X1ZY} = X_1 Z_3 Y_4 (= X\otimes 1\otimes Z\otimes Y)$.
Given $q$ obtained by evolving $q_0$ by a unitary $U$ (i.e.\ 
$q = U^\dagger q_0 U$), extend each $q(j)$ by 
$q(j)_Y := i q(j)_X q(j)_Z$ and $q(j)_1 := 1$, where $i$ is the 
imaginary unit.
Then the density matrix of $U\ket{0^n}$, the state encoded in $q$, 
is given by the expectation-weighted sum formula as:
\begin{equation}\label{eq:dh-density-whole}
  \density q = \sum_{s \in S^n}\bra{0^n}U^\dagger P_s U\ket{0^n} P_s 
  = \sum_{s \in S^n}\left\langle0^n \middle| \prod_{j=1}^n q(j)_{s_j} \middle| 0^n \right\rangle P_s.
\end{equation}

There are several $n$'s in this formula, but they are not the same 
$n$: the $n$ on top of $\prod$ and the one in $S^n$ are $\#\dom q$, 
the number of qubits that $q$ prescribes descriptors for, but the 
other $n$'s are the number of qubits that each descriptor acts on, 
i.e.\ the number of qubits that are entangled with these descriptors.
This difference becomes consequential when we restrict $\dom q$ 
without restricting the qubits the descriptors act on.
\begin{equation}\label{eq:dh-density}
  \density (q|_{\{1..k\}}) = 
  \sum_{s \in S^{k}}\left\langle0^n \middle| \prod_{j=1}^{k} q(j)_{s_j} 
    \middle| 0^n \right\rangle P_s.
\end{equation}

This formula is exactly what we get if we discard (i.e.\ trace out) 
all but the first $k$ qubits from the whole-system density matrix.
Taking a partial trace of \cref{eq:dh-density-whole} gives
\begin{equation*}
  \tr_{k+1..n} (\density q) 
  = \sum_{s \in S^{k}, s' \in S^{n-k}}\left\langle0^n \middle| \prod_{j=1}^{k} q(j)_{s_j} 
    \cdot \prod_{j=1}^{n-k} q(j)_{s'_j} 
    \middle| 0^n \right\rangle P_s \cdot \tr P_{s'},
\end{equation*}
but $X,Y,Z$ are traceless, so the only terms that survive are those 
with $s' = 1^{n-k}$.
This formula then reduces to exactly the right-hand side of 
\cref{eq:dh-density}, hence 
$\tr_{k+1..n}(\density q) = \density(q_{\{1..k\}})$.
Thus, the density matrix of a subset of qubits is just the density 
matrix of the subset of DH descriptors, where individual descriptors 
still act on the whole set of entangled qubits.

\section{Splittable, Mergeable Descriptors}
\label{sec:splittable-dh}

This section develops the mathematical machinery needed to turn DH 
descriptors into a spatially compositional representation.
The idea is to stop identifying qubits by position and to identify 
them by name instead, so that each qubit has a stable ID that stays 
meaningful when they are completely separated from each other.

DH descriptors as explained in the preceding section are not 
compositional because each individual descriptor $q(j)_X$ or $q(j)_Z$ 
is a \emph{global} operator acting on all qubits in the system.
This monolithic formulation is more or less forced by the choice of 
standard DH to identify qubits by their position on a tensor product.
Effectively, it is identifying qubits by numerical indices into the 
global store, which is incompatible with splitting and merging.
For instance, splitting should re-index qubits so that $\dom q$ 
remains of the form $\{1..n\}$; yet if each subset re-indexed 
independently, then a descriptor in subset $q_1$ may re-index to 
qubit 5 what subset $q_2$ re-indexed to qubit 7, confusing the 
representation of entanglement.

The solution is to identify qubits by names instead of 
context-dependent indices.
The idea is to redefine DH descriptor sets from 
$q : \{1..n\} \rightarrow \Obsvs_n^2$ to 
$q : V \rightarrow \Obsvs_W^2$, where $V,W$ are finite subsets of a 
fixed infinite supply of variables $\Var$, and $\Obsvs_W$ is the set 
of Hermitian operators on $\mathbb{C}^{2^W}$.
Effectively, we replace the $n$ in $\mathbb{C}^{2^n}$, the usual 
domain of $n$-qubit system states, with $V$.

How much does this replacement alter the linear algebra?
Not much.
A conventional vector $v \in \mathbb{C}^{2^n}$ is a list of numbers 
indexed by $n$-bit wide binary numbers 
$v = [v_0, v_1, \ldots, v_{2^n-1}]$.
A vector $v \in \mathbb{C}^{2^V}$ is a collection of numbers indexed 
by valuations on $V$, i.e.\ maps $V \rightarrow 2$.
If we fix an ordering on $V = \{x_0, x_1, \ldots, x_{n-1}\}$, then we 
can identify the valuation 
$[x_{n-1} \mapsto b_{n-1}, x_{n-2} \mapsto b_{n-2}, \ldots, x_{0} 
\mapsto b_0]$ with the bit string $b_{n-1}b_{n-2}\ldots b_0$ and the 
natural number it represents.
Then $v$ is once again $v = [v_0, v_1, \ldots, v_{2^{n-1}}]$, except 
each natural number $k \leq 2^{n-1}$ is understood as a valuation 
through the bit string encoding.

Likewise, matrices on $\mathbb{C}^{2^V}$ are just arrays of numbers 
$A_{\eta\theta}$, where the row and column indices are valuations 
$\eta,\theta : V \rightarrow 2$.
Hermitian conjugation and Hermiticity are defined as usual by 
$(A^\dagger)_{\eta\theta} = (A_{\theta\eta})^*$ and $A^\dagger = A$, 
respectively.
A Pauli $X$ on a single-qubit space $\mathbb{C}^{2^{\{x\}}}$ is just 
the matrix $X_x$ with the $([x \mapsto 0],[x \mapsto 0])$ and 
$([x \mapsto 1],[x \mapsto 1])$ entries 1, and the two other entries 
0.

The one thing that does change is that each qubit now has a natural 
identity, independent of what other qubits are at play.
There is a single, canonical interpretation for $X_x$, be it on 
$\mathbb{C}^{2^{\{x\}}}$, $\mathbb{C}^{2^{\{x,y,z\}}}$, or 
$\mathbb{C}^{2^V}$ for any $V$, as long as $x \in V$.
This independent identity makes re-indexing unnecessary.
For example, the store
\begin{equation*}
  \rho := [x \mapsto (Z_x X_y, X_x), y \mapsto (X_y, X_x Z_y)]
\end{equation*}
with two entangled qubits can be split up into
\begin{mathpar}
  \rho|_{\{x\}} = [x \mapsto (Z_x X_y, X_x)] \and \text{and} \and 
  \rho|_{\{y\}} = [y \mapsto (X_y, X_x Z_y)],
\end{mathpar}
and each piece makes perfect sense without the other.
Sub-store $\rho|_{\{x\}}$ knows that $x$ is a qubit entangled with 
\emph{some} qubit named $y$ in the indicated manner, and what state 
that $y$ is in is not $\rho|_{\{x\}}$'s concern.
Likewise, $\rho|_{\{y\}}$ knows $y$ is entangled with \emph{some} 
external qubit named $x$ somewhere out there.
Thanks to locality (\cref{thm:dh-modularity}), operations on qubit 
$x$ need only modify $\rho|_{\{x\}}$.
After independent evolution, the fragments can be put back together 
at any time by $\rho = \rho|_{\{x\}} \cup \rho|_{\{y\}}$.

One lingering issue is that, technically, $X_x$ is still a different 
matrix in $\Obsvs_V$ as it is in $\Obsvs_W$ (assuming $V \neq W$ but 
both contain $x$), so a reinterpretation is still needed when 
splitting or merging changes the variable set we are working with.
To avoid constant reinterpretation, we need a domain of 
representations that embeds $\Obsvs_V$ for all 
$V \subseteq_\fin \Var$, which we now construct.
For readers familiar with universal algebra, what follows is 
essentially a construction of the intended model for the algebra of 
ring expressions in $X_x$, $Z_x$ for all $x \in \Var$ modulo 
equalities that hold when the $X_x$'s and $Z_x$'s are interpreted in 
a finitary model $\mathbb{C}^{2^V\times 2^V}$.
The constructed model is a colimit of those finitary models along 
with evident embeddings 
$\mathbb{C}^{2^V\times 2^V} \hookrightarrow \mathbb{C}^{2^W\times 
  2^W}$ where $V \subseteq W$.

\begin{definition}
  For a set $S$, let 
  $\mathbb{C}^{(S)} := \bigoplus_{\_ \in S} \mathbb{C}$ i.e.\ the set 
  of all finite linear combinations of elements of $S$.
  Define its inner product by 
  $\langle u, v\rangle := \sum_{s \in S} u_s^* v_s$, which is a 
  finite sum and therefore always defined.
\end{definition}

\noindent Note that for finite $V$, we have 
$\mathbb{C}^{(2^{\Var})} \cong \mathbb{C}^{(2^V \times 2^{\Var-V})} 
\cong \mathbb{C}^{(2^V)} \otimes \mathbb{C}^{(2^{\Var-V})} = 
\mathbb{C}^{2^V} \otimes \mathbb{C}^{(2^{\Var-V})}$.

\begin{definition}
  For $V \subseteq_\fin \Var$, define $\iota_V$ to be the map taking 
  $A \in \mathbb{C}^{2^V\times 2^V}$ to the linear endomorphism 
  $A\otimes 1_{\mathbb{C}^{\left(2^{\Var-V}\right)}}$ on 
  $\mathbb{C}^{(2^\Var)}$.
  Let $\Domain := \bigcup_V \img \iota_V$.
  Define Hermitian conjugation on $\delta \in \Domain$ as usual by 
  $\langle \delta^\dagger u, v \rangle = \langle u, \delta v 
  \rangle$.
  The set of \emph{nominal DH descriptors} is 
  $\Obsvs := \{\delta \in \Domain \mid \delta^\dagger = \delta\}$.
  The \emph{support} of $\delta \in \Domain$, written $\supp \delta$, 
  is the set of variables on which $\delta$ has nontrivial action, 
  i.e.\ the minimal $V$ such that $\delta \in \img \iota_V$.
\end{definition}

Intuitively, the support of a descriptor is the set of variables that 
``appear'' in it.
For example, $\supp(Z_xX_y) = \{x,y\}$; however, the support excludes 
variables that can be simplified away, e.g.\ 
$\supp (Z_xZ_x) = \supp 1 = \varnothing$.

Each $\iota_V$ is an embedding 
$\mathbb{C}^{2^V\times 2^V} \hookrightarrow \Domain$, ensuring every 
$\delta \in \Domain$ can be manipulated as if it's an element of 
$\mathbb{C}^{2^V\times 2^V}$ for sufficiently large $V$.
For example, $X_x Z_y = Z_y X_x$ in $\Domain$ because 
$X_x Z_y = Z_y X_x$ in $\mathbb{C}^{2^{\{x,y\}}\times 2^{\{x,y\}}}$.
One can check that Hermitian conjugation also transfers along the 
$\iota$'s; that is, $(\iota_V A)^\dagger := \iota_V(A^\dagger)$ 
satisfies the definition of Hermitian conjugation on $\Domain$.
Therefore, the embedded manipulation extends to ${}^\dagger$.
Note $\Domain$ is the union of images of 
$\mathbb{C}^{2^V\times 2^V}$'s, so every $\delta \in \Domain$ has a 
finite multilinear expression in the $X_x$'s and $Z_x$'s that can be 
manipulated this way.

Now we can define the central data structure of our semantics, the 
store, which captures the state of quantum data in the 
calculus:

\begin{definition}
  A \emph{store} is a mapping $\rho : V \rightarrow \Obsvs^2$ where 
  $V \subseteq_\fin \Var$.
  Its \emph{support} is 
  $\supp \rho := \dom \rho \cup \bigcup_{x \in \dom \rho} \supp 
  (\rho(x))$.
  A store is \emph{full} when $\supp\rho = \dom \rho$, \emph{partial} 
  otherwise.
\end{definition}

The $V$ is the set of qubit names that $\rho$ assigns states to.
Individual nominal DH descriptors in $\img \rho$ may mention 
variables outside of $V$, suggesting entanglement with the outside 
world.
Nominal DH descriptors can be evolved exactly like ordinary DH 
descriptors, except the $X_x$'s and $Z_x$'s in their algebraic 
expressions identify qubits by names instead of numerical indices.

Thus, a store can be manipulated just like ordinary DH descriptor 
sets, except they can be split and merged with impunity.
For example, if $V$ partitions as $V = V_1\uplus V_2$, then $\rho$ 
can be split into parts $\rho_1$ and $\rho_2$ which are just domain 
restrictions $\rho_j = \rho|_{V_j}$.
Likewise, if $\rho$ is merged with another store 
$\mu : W \rightarrow \Obsvs^2$ (where $V \cap W = \varnothing$), the 
result is $\rho \cup \mu$ on the nose.

A density matrix can be recovered from such a store by adapting 
\cref{eq:dh-density}:
\begin{equation}\label{eq:nominal-dh-density}
  \density \rho
  = \sum_{s \in S^V}\biggl\langle 0^W \biggm| \prod_{x \in V} \rho(x)_{s_x} 
    \biggm| 0^W \biggr\rangle P_s
\end{equation}
where $P_s := \prod_{x \in V} (s_x)_x$ as before, $V := \dom \rho$, 
$W := \supp \rho$, and $0^W$ denotes the valuation sending everything 
in $W$ to 0.
Note $\density \rho \in \mathbb{C}^{2^V\times 2^V}$.
Each $\rho(x)_{s_x}$ is interpreted in $\mathbb{C}^{2^W\times 2^W}$ 
through the partial inverse of $\iota_W$, which is guaranteed to be 
defined on $\rho(x)_{s_x}$ due to the definition of $W$.

\section{A Spatially Compositional Process Calculus}
\label{sec:dh-ccs}

\begin{figure}[t]
  \begin{mathpar}
    c,d \in \textit{Chan} \and
    L \subseteq \textit{Chan} \and
    x,y,z \in \Var \and
    V,W \subseteq \Var \and
    q,r : \{x\} \rightarrow \Obsvs^2 \and
    \rho,\mu : V \rightarrow \Obsvs^2
  \end{mathpar}
  \[
    \begin{array}{rcl@{}llcl}
      a \in \textit{Action} & ::= & c!q \mid c?q \mid \tau \\
      P,Q \in \Proc & ::= &
      \multicolumn{5}{l}{
      0
      \mid \tau.P
      \mid P \setminus L
      \mid A(\tilde x)
      \mid P + Q\ [\FV P = \FV Q]}
      \\&\mid &
      c!x.P & [x \notin \FV P] &\mid& !!x.P & [x \notin \FV P]
      \\&\mid &
      c?x.P & [x \in \FV P] &\mid& ??x.P & [x \in \FV P]
      \\&\mid &
      U(\tilde x).P & [\{\tilde x\} \subseteq \FV P]
      &\mid& P \parcomp Q & [\FV P \cap \FV Q = \varnothing]
      \end{array}
    \]
    \begin{align*}
      \FV (P \parcomp Q) &= \FV P \cup \FV Q
      & \FV (P \setminus L) &= \FV P
      & \FV (A(\tilde x)) &= \{\tilde x\}
      & \Ch(c?q) &= c \\
      \FV (P + Q) &= \FV P \cup \FV Q
      & \FV (c!x. P) &= \{x\} \cup \FV P
      & \FV (!!x. P) &= \{x\} \cup \FV P
      & \Ch(c!q) &= c \\
      \FV (c?x. P) &= \FV P - \{x\}
      & \FV (??x. P) &= \FV P - \{x\}
      & \FV (U(\tilde x). P) &= \{\tilde x\} \cup \FV P
      & \Ch \tau &= \bot \\
      \FV 0 &= \varnothing  & \FV (\tau. P) &= \FV P &
    \end{align*}
    \caption{Syntax of DH-CCS.
      Production rules marked with $[\phi]$ apply only if $\phi$ is 
      true.}
  \label{fig:syntax}
\end{figure}

This section gives a spatially compositional quantum process 
calculus, DH-CCS, using nominal DH descriptors.
The design largely follows qCCS 
\cite{FengEtAl2007ProbabilisticBisimulationsQuantum} and lqCCS 
\cite{CeragioliEtAl2024QuantumBisimilarityBarbs}.

\Cref{fig:syntax} gives the syntax.
We have a stuck process $0$, silent prefix $\tau. P$, hiding 
$P \setminus L$, named recursive processes $A(\tilde{x})$, choice 
$P+Q$, send $c!x. P$, receive $c?x. P$, discarding $!!x.P$, 
allocation $??x.P$, local unitaries $U(\tilde x)$, and parallel 
composition $P \parcomp Q$.
We assume there are definitions for recursive processes of the form 
$A(\tilde x) := P$ where $\tilde x = \FV P$, and where $P$ may 
mention $A$ or other named recursive processes.
Every variable is of qubit type.
We leave out classical values and classical control (i.e.\ 
\texttt{if}-\texttt{then}-\texttt{else}), but measurement can be 
implemented by discarding; see \cref{sec:unification}.

We impose linearity constraints similar to lqCCS, but without a type 
system.
Each variable is thought of as a qubit with a definite physical 
identity, which cannot be in multiple places at once, so 
$P\parcomp Q$ requires $P$ and $Q$ to partition the set of available 
variables, while $c!x. P$ requires $x$ to never be used again, i.e.\ 
$x \notin \FV P$.
Reversibility means that qubits cannot be freely dropped, so $c?x. P$ 
requires that the received qubit be used somewhere, hence 
$x \in \FV P$.
Similarly, $U(\tilde{x})$ requires the $x$'s to be used afterwards, 
while $A(\tilde x) := P$ requires that every $x$ appear in $P$, even 
if it's just to be passed around to a recursive call.

\newif{\ifcong}
\congtrue
\begin{figure*}[t]
  \begin{mathpar}
    \inferrule[Send]{\rho|_{\{x\}} = q} 
    {\st{c!x.P}{\rho} \sstep{c!q} \st{P}{\rho|_{\dom \rho - \{x\}}}}
    \and
    \inferrule[Recv]{\dom q = \{y\} \and y 
      \notin \dom \rho} {\st{c?x.P}{\rho} \sstep{c?q} 
      \st{[y/x]P}{\rho\cup q}}
    \and
    \inferrule[Comm]{\st{P}{\rho|_{\FV P}} \sstep{c!q} 
      \st{P'}{\rho'_P} \and \st{Q}{\rho|_{\FV Q}} \sstep{c?q} 
      \st{Q'}{\rho_Q'}} {\st{P \parcomp Q}{\rho} \sstep{\tau} \st{P' 
        \parcomp Q'}{\rho_P'\cup\rho_Q'}}
    \and
    \inferrule[Discard]{ }{\st{!!x.P}{\rho} 
      \sstep{\tau} \st{P}{\rho|_{\dom \rho - \{x\}}}}
    \and
    \inferrule[Alloc]{x\text{ fresh}}{\st{??x.P}{\rho} 
      \sstep{\tau} \st{P}{\rho \cup \rho_0(x)}}
    \and
    \inferrule[unitary]{\phantom{U}}{\st{U(\tilde{x}).P}{\rho} 
      \sstep{\tau} \st{P}{f_{U(\tilde{x})}(\rho)^\dagger \, \rho\, 
        f_{U(\tilde{x})}(\rho)}}
    \and
    \inferrule[CongSend]{\st{P}{\rho|_{\FV P}} \sstep{c!q} \st{P'}{\rho'_P}} 
    {\st{P \parcomp Q}{\rho} \sstep{c!q} \st{P' \parcomp Q}{\rho'_P 
        \cup \rho|_{\FV Q}}}
    \and
    \inferrule[CongRecv]{\st{P}{\rho|_{\FV P}} \sstep{c?q} \st{P'}{\rho'_P} 
      \and \dom q \cap \dom \rho = \varnothing} {\st{P \parcomp 
        Q}{\rho} \sstep{c?q} \st{P' \parcomp Q}{\rho'_P \cup 
        \rho|_{\FV Q}}}
    \and
    \inferrule{ }
    {\st{\tau.P}{\rho} \sstep{\tau} \st{P}{\rho'}}
    \and
    \inferrule{\st{P}{\rho} \sstep{a} \st{P'}{\rho'}}{\st{P+Q}{\rho} 
      \sstep{a} \st{P'}{\rho'}}
    \and
    \inferrule{\st{P}{\rho} \sstep{a} \st{P'}{\rho'} \and \Ch a 
      \notin L}{\st{P \setminus L}{\rho} \sstep{a} \st{P \setminus 
        L}{\rho'}}
    \and
    \inferrule{A(\tilde{x}) := P}{\st{A(\tilde{y})}{\rho} 
      \sstep{\tau} \st{[\tilde{y}/\tilde{x}]P}{\rho}}
    \ifcong
      \and
      \inferrule{\st{P}{\rho|_{\FV P}} \sstep{\tau} \st{P'}{\rho'_P}}{\st{P 
          \parcomp Q}{\rho} \sstep{\tau} \st{P' 
          \parcomp Q}{\rho_P' \cup \rho|_{\FV Q}}}
    \fi
  \end{mathpar}
  \caption{Semantics for DH-CCS.
    Mirror-image rules for $\parcomp$ and $+$ are omitted.
    \ifcong\else Congruence rules are written only where they need 
      special care.\fi }
  \label{fig:semantics}
\end{figure*}

\Cref{fig:semantics} gives the semantics, defining a transition 
relation between configurations of the form $\st{P}{\rho}$ where 
$\rho : \FV P \rightarrow \Obsvs$.
Execution of a process term $P$ begins with $\st{P}{\rho_0(\FV P)}$ 
where $\rho_0(V)$, for any $V \subseteq \Var$, is the all-zero store 
that assigns to every $x \in V$ the descriptors $(X_x, Z_x)$ 
representing the pure state $\ket{0}$.

Send $\st{c!x.P}{\rho}$ works by splitting $\rho$ into the sub-store 
$q$ holding just the qubit named $x$ and the rest $\rho|_{\FV P}$, 
then sending out the $q$ in a $c!q$ message.
This is quite unusual---most quantum process calculi only communicate 
qubit names, not qubit states.
Receive $\st{c?x.P}{\rho}$ accepts a message $c?q$ carrying a qubit 
and substitutes its name for $x$; the $q$, being a store itself, is 
simply merged into $\rho$.
The $c!q$ and $c?q$ messages meet and annihilate into $\tau$ at 
$\parcomp$ in the standard manner.

Allocation $??x.P$ works just like $c?x.P$, except it works without 
any incoming message; intuitively, it receives a $??\rho_0(y)$ 
message with fresh $y$ from the environment over a dedicated, hidden 
read-only channel named ``$?$''.
Likewise, discarding $!!x.P$ is similar to $c!x.P$ over a dedicated, 
hidden write-only channel named ``$!$''.
A unitary $U(\tilde{x})$ is converted to its functional expansion 
$f_{U(\tilde{x})}$ (from \cref{sec:dh}) and applied to the store.


The other rules are mostly standard, but note how propagation rules 
for $c!q$ and $c?q$ must take care to observe linearity.
In particular, \textsc{CongRecv} requires that the in-flight qubit 
$q$ be distinct from those in the store $\rho$.
Since qubit identities are tied to names, we check this condition by 
checking the disjointness of $\dom q$ and $\dom \rho$.
A similar disjointness check is technically also needed for 
\textsc{Comm} and \textsc{CongSend} (in both cases 
$\dom q \cap \FV Q = \varnothing$), but there the checks are subsumed 
by the well-formedness condition for $P \parcomp Q$, as $P$ can only 
emit $c!q$ if $\dom q \subseteq \FV P$, which forces 
$\dom q \cap \FV Q = \varnothing$.
The check is not automatic for receiving because the name of an 
incoming qubit is chosen by the sender, not the local process.

The calculus does not require stores to satisfy the Pauli equations 
from \cref{defn:su2-algebra} that characterize physically valid DH 
descriptor sets.
The transition rules barely look at the contents of the descriptors, 
merely pushing around fragments of the store, so it works fine with 
unphysical descriptors.
Store validity becomes more interesting when we consider process 
equivalence, which hinges on the contents of the descriptors; see 
\cref{sec:bisimulation}.

As one would expect, qubits can be freely renamed in this calculus 
without affecting the semantics.
We follow Gabbay and Pitts \cite{GabbayPitts2002NewApproachAbstract} 
and represent renaming schemes by permutations rather than arbitrary 
functions.\footnote{We also borrowed the term ``support'' from them, 
  but our definition of $\supp \rho$ is not an instance of theirs, as 
  we preferred a definition that applies more directly to our 
  situation.}
In their terminology, DH-CCS's semantics is \emph{equivariant}.

\begin{definition}
  Given a permutation $\sigma : \Var \cong \Var$, define its actions 
  on various data as follows:
  \begin{itemize}
  \item $\sigma V := \sigma[V]$ for $V \subseteq \Var$.
  \item $\sigma P$ permutes all names in $P$, bound or free.
  \item $\sigma z = z$ for complex numbers $z$.
    This includes the Booleans $0,1$.
  \item $\sigma f := \sigma \circ f \circ \sigma^{-1}$ where $f$ is a 
    valuation, vector, linear map, or nominal DH descriptor, all seen 
    as functions $A \rightarrow B$ between sets $A$, $B$ on which 
    $\sigma$ action is defined.
  \item $\sigma (c!q) := c!(\sigma q)$ and 
    $\sigma (c?q) := c?(\sigma q)$ and $\sigma \tau := \tau$ for 
    actions.
  \item 
    $\sigma (x_1, \ldots, x_n) := (\sigma x_1, \ldots, \sigma x_n)$ 
    for any tuple of $x_j$'s on which $\sigma$ action is defined.
    This is a special case of functions (with $A = \{1,\ldots, n\}$) 
    and includes configurations and traces.
  \end{itemize}
  A function between objects with $\sigma$-action defined is 
  \emph{equivariant} iff it commutes with $\sigma$.
  A predicate is equivariant iff it is preserved by application of 
  $\sigma$ to its arguments.
\end{definition}

\begin{example}
  Let $\sigma$ swap $x,y$ while fixing everything else.
  Its action on:
  \begin{itemize}
  \item A Pauli symbol $X_x$ or $Z_x$ changes the subscript: 
    $\sigma X_x = X_y$, $\sigma Z_y = Z_x$.
  \item A nominal DH descriptor set 
    $[x \mapsto (X_x, Z_xX_y), y \mapsto (X_y, X_xZ_y)]$ swaps the 
    visible $x$'s and $y$'s: 
    $\sigma[x \mapsto (X_x, Z_xX_y), y \mapsto (X_y, X_xZ_y)] = [y 
    \mapsto (X_y, Z_yX_x), x \mapsto (X_x, X_yZ_x)]$.
    \qedhere
  \end{itemize}
\end{example}

\noindent Note that $\sigma$ preserves $\Domain$: it maps one finite 
multilinear combination of Pauli symbols to another.

\begin{proposition}\label{thm:equivariance}
  If $\st{P}{\rho} \sstep{a} \st{Q}{\mu}$, then 
  $\sigma\st{P}{\rho} \sstep{\sigma a} \sigma\st{Q}{\mu}$.
  \begin{proof}
    Induction on the derivation, using and proving equivariance.
    For example, if 
    $\st{c!x.P}{\rho} \sstep{a} \st{P}{\rho|_{\dom \rho - \{x\}}}$ is 
    derived by the \textsc{Send} rule, then inversion gives 
    $\rho|_{\{x\}} = q$.
    Applying $\sigma$ to both sides gives 
    $\sigma\rho|_{\{x\}} = \sigma q$.
    In a straightforward lemma, one can prove that domain restriction 
    is equivariant, so the $\sigma$ can be pushed inward, giving 
    $(\sigma\rho)|_{\{\sigma x\}} = \sigma q$.
    Applying \textsc{Send} to this fact gives 
    $\st{c!(\sigma x).(\sigma P)}{\sigma \rho} \sstep{c!(\sigma q)} 
    \st{\sigma P}{(\sigma \rho)|_{\dom (\sigma \rho) - \{\sigma 
        x\}}}$, and equivariance allows us to pull out the $\sigma$.
  \end{proof}
\end{proposition}

\section{Boundaries of Communication}
\label{sec:measurement}\label{sec:unification}

In this section, we explain how DH-CCS unifies allocation, 
discarding, and the probability-shifting aspects of measurement with 
communication.
These operations are all viewed as communication with the 
environment, hence as transfer of qubits across the system boundary.
Classical conditionals branching on measured values lie outside of 
this unification because it entails a change of perspective.

The semantics for allocation $??x.P$ and discarding $!!x.P$ are 
nearly identical to those of receiving $c?x.P$ and sending $c!x.P$, 
respectively.
This is not a coincidence.
In real-life quantum computers, allocation is the act of snatching an 
already-existing qubit on the machine and declaring it a part of the 
allocating process, just like in classical computing.
Discarding is a two-step process, where the quantum computer measures 
the qubit and then declares it no longer part of the process.
The commonality between them and communication is that all three are 
\emph{transfers of ownership} at heart.

Allocation and discarding can be seen as communication with the 
surrounding environment, transferring ownership across the 
process-environment boundary.
If we postulate a process named $\env$ which can transition in 
exactly two ways, $\env \sstep{!?q} \env$ and 
$\env \sstep{?!\rho_0(x)} \env$ for fresh $x$, then in any process 
$P$ appearing as $(\env \parcomp P) \setminus \{!,?\}$, the 
constructs $!!x.Q$ and $??x.Q$ become special cases of $c!x.Q$ and 
$c?x.Q$, respectively.
Even the measurement inherent in discarding is an artifact of 
ownership transfer.
Say we discard $x \in V$ from $\rho : V \rightarrow \Obsvs$.
Then the density matrix of the remaining qubits $V-\{x\}$ can be 
reconstructed from $\rho|_{V-\{x\}}$ using 
\cref{eq:nominal-dh-density}, and this is identical to the density 
matrix corresponding to $\rho$ with $x$ traced out, as shown in 
\cref{sec:dh}.

We chose not to explicitly present allocation and discarding as 
special cases of communication in our definition of DH-CCS because it 
introduces other technicalities:
\begin{itemize}
\item $\env$ is an unbounded process, doling out an unlimited supply 
  of fresh qubits, so we would need to either postulate it as a 
  special process or extend the syntax of recursive processes 
  $A(\tilde{x})$ to allow infinite $\tilde{x}$, with knock-on effects 
  on when fresh variables can be guaranteed.
\item The ordering of qubit discarding should be immaterial, but this 
  would not be true if one is allowed to observe the ordering of 
  $!!q$ actions.
  The $(\env \parcomp P) \setminus \{!,?\}$ combination ensures 
  order-insensitivity because $\env$ absorbs ordering differences and 
  $\setminus \{!,?\}$ ensures those messages don't escape further 
  out.
  To accurately reflect this fact, the definition of observation will 
  have to mimic the setup $(\env \parcomp P) \setminus \{!,?\}$, 
  requiring it to be aware of $\env$ and what it does.
\item There must be a notion of read-only and write-only channels to 
  prevent the $!$ and $?$ channels from being used for communication 
  within the main process.
\end{itemize}
These appear to be mere nuisances rather than insurmountable 
difficulties, but it's simpler just to treat allocation and 
discarding as built-in.

Nevertheless, DH-CCS contributes something new to the handling of 
allocation, discarding, and communication: it manages to formally 
unify the semantic action modeling ownership transfer in all three 
cases.
If a process sheds a qubit, its store is restricted as $\rho|_V$, and 
if a process gains a qubit, its store is absorbs the qubit as 
$\rho\cup q$---it doesn't matter \emph{why} those transfers happened, 
it's always the same underlying operation.
By contrast, all existing semantics we are aware of introduces a 
tensor product upon allocation while treating internal communication 
as a no-op on the store, only updating the ownership ledger of 
qubits.\footnote{The store operation of internal communication in 
  DH-CCS also amounts to a no-op, but it is the result of composing 
  the split and merge operations described here.}
Discarding is even more radically different: it introduces ensembles, 
or classical probability distributions over configurations.

The difference between communication and discarding is that the 
former can be reversed while the latter generally comes with a 
promise never to reverse it.
Ensembles and partial traces commit to this irreversibility, 
forgetting how a system was entangled with the discarded qubit.
This is why they are incomplete: having ensemble information or 
partial traces for all processes in a system does not allow to 
reconstruct the global state.

But this promise is just a convention, which need not be represented 
in the state.
Without the promise, discarding is just another communication.
DH descriptors can model ownership transfer in a way that is 
non-committal about whether that transfer may be reversed later on, 
thereby unifying discarding with communication.

Measurement is another form of communication: in the simplest case, 
it transfers an entangled copy\footnote{Not to be confused with a 
  clone, forbidden by no-cloning.
  By entangled copy, we mean the other end of a Bell pair: it always 
  gives the same measurement result as the original, whereas a clone 
  replicates the probability distribution of the original while 
  allowing individual trials to diverge.} of a qubit in the measured 
system to the measuring apparatus as a recording of the measurement 
outcome.
More generally, the apparatus can extract partial correlations with 
the system, but the implementation is always that an entangled qubit 
(or other unit of quantum data) moves from the system to the 
apparatus.
The apparatus is normally considered a part of the environment, so 
this too is a form of communication with the environment.

Communication with the environment is a quintessential open-system 
operation.
Existing semantics have trouble unifying these actions because they 
lack a clean way to model interactions across system boundaries 
whenever entanglement is a possibility.

Alas, while the way measurement shifts quantum uncertainty from one 
locus to another is a form of ownership transfer, the way that 
measured bits can be used for classical control lies outside this 
semantic unification.
Quantum computation has two notions of control, or conditional 
execution.
Quantum control, a.k.a.\ controlled unitary, branches on a qubit and 
executes a quantum superposition of the conditional operation and a 
no-op.
Classical control branches on a classical value and performs or calls 
off an operation entirely, creating a classical superposition of 
configurations known as an ensemble.
It is the creation of ensembles that cannot be modeled with ownership 
transfer.

The difference between quantum and classical control is a matter of 
perspective.
If we consider the quantum computer, its measurement apparatus, and 
its classical controller as a single quantum system, its evolution is 
still unitary.
Measurement is just a unitary operation that correlates the value of 
a qubit with the measurement apparatus' supposedly classical 
register---which is still a quantum system, as anything is.
Branching on this register is then understood as a controlled 
unitary, with the register as the control.

Classical branching appears only when we stop modeling the controller 
as part of the system and instead condition on the value recorded in 
it \cite{Everett1957RelativeStateFormulation}.
We decompose the quantum computer's state into so-called relative 
states indexed by the controller state and reason about each branch 
separately.
Weighting these branches by the probabilities of finding the 
controller in the corresponding state yields an ensemble.
Thus, classical control does not arise from a transfer of qubit 
ownership but from a change in description, where we identify with 
the controller and think of its register's value as defining part of 
our worldview.
This change in perspective is the part that does not fall under 
ownership transfer.

We leave classical control out of this paper because it appears to be 
a largely orthogonal extension to this paper's main concern of 
spatial compositionality.
One could add classical control to DH-CCS in the conventional way, 
modifying it to have a transition system over ensembles of 
configurations.
The ownership-transfer interpretation of allocation, discarding, 
communication, and measurement survives unchanged; it merely has to 
be applied separately to each branch of the ensemble.
Spatial compositionality would similarly assert that global 
evolution, weighted by probability of realization, can be pieced 
together from probability-weighted local evolution.

Thus, classical control seems to substantially complicate the 
presentation without adding much insight.
While it could be worthwhile to look for confirmation that classical 
control can \emph{really} be added to DH-CCS via ensembles, a more 
interesting question in our opinion would be whether, say, the more 
fluid handling of splitting enabled by DH descriptors provides an 
alternative understanding of conditioning to ensembles.
In any case, that warrants a separate paper.

One thing the reader may worry about is whether critical expressivity 
is lost by leaving out classical data and control.
We think not.
\begin{itemize}
\item Communication over classical channels is just quantum 
  communication where an entangled copy is sent off to the 
  environment.
  More formally, a channel $c$ is classical iff every send operation 
  on it has the form $??t.\,\CX(x,t).\,!!t.\,c!x.P$.
\item Unitaries under classical conditionals can be systematically 
  replaced by controlled unitaries, where the classical bits 
  controlling them are just qubits that happen to have entangled 
  copies in the environment.
\item Conditional sending of qubits depending on measurement outcome 
  is simulated by replacing conditional transfer of \emph{ownership} 
  with conditional transfer of \emph{payload}.
  Concretely, we do 
  $??y.\, \operatorname{CSWAP}(b,x,y).\,c!b.\,c!x.\,!!y.\,P$, where 
  $\operatorname{CSWAP}$ is the controlled-swap gate, which swaps 
  $x,y$ subject to $b$ being $1$.
  When $b=0$, instead of not sending anything, we send out a dummy.
  The $b$ itself is also sent out to tell the receiver whether the 
  payload is genuine or a dummy.
\item Conditional receive is simulated by unconditional reception 
  followed by discarding conditioned (with the same technique as 
  conditional send) on that $b$ bit which tells whether the payload 
  is genuine.
\end{itemize}
One thing that cannot be simulated well is dynamic selection of which 
channel to send on.
Channel selection can affect scheduling, which is a classical part of 
the model, so conditioning that on qubit values requires classical 
control.
But plain CCS also has no dynamic channel selection and is plenty 
useful; we believe the same holds for DH-CCS as presented here.

\section{Spatial Compositionality}
\label{sec:metatheory}

In this section, we define what spatial compositionality means and 
discuss why it is so difficult to ensure with density matrices.
We see that DH descriptors enable spatial compositionality because it 
contains compressed encodings of evolution history.

Spatial compositionality means that each global configuration 
$\st{P\parcomp Q}{\rho}$ defines process-local views $\st{P}{\rho_P}$ 
and $\st{Q}{\rho_Q}$, which, after independent evolution, can recover 
the global evolved configuration $\st{P' \parcomp Q'}{\rho'}$.
To understand how unique this property is, let us imagine what our 
options would be if the $\rho$ were a density matrix.
The only generally accepted way to restrict a density matrix to a 
subset of the qubits is the partial trace, but partial traces are 
inadequate as local views because they throw away information about 
entanglement with the traced out qubits.

\begin{definition}
  Let us write $\Phi^\pm$ to denote the two Bell states 
  $\Phi^\pm := \frac{1}{\sqrt{2}}(\ket{00} \pm \ket{11})$ and write 
  $\Bell^\pm(x,y).P$ for the processes that prepare $x,y$ in 
  respectively the $\Phi^{\pm}$ state before acting as $P$:
  \begin{align*}
    \Bell^+(x,y).P &:= H(x).\CX(x \rightarrow y).P \\
    \Bell^-(x,y).P &:= H(x).\CX(x \rightarrow y).Z(x).P.
  \end{align*}
  A \emph{trace} is a sequence of actions, denoted by metavariables 
  $t,s$.
  Let $\sstep[*]{t}$ denote zero or more transitions emitting $t$ 
  along the way.
  We may omit the $t$ when it's unimportant.
\end{definition}

\begin{example}\label{eg:partial-trace-problem}
  Let us show the information loss of partial traces.
  Let $P^{\pm} := \Bell^\pm(x,y).(A(x) \parcomp B(y))$, where $A(x)$ 
  and $B(x)$ are some named recursive processes.
  Letting 
  $\st{P^\pm}{\rho_0(x,y)} \sstep[*]{} \st{A(x) \parcomp 
    B(y)}{\rho^\pm}$, what should be the views of $\rho^\pm$ local to 
  the subprocess $A(x)$, denoted $\rho^\pm_A$?
  What about $\rho^\pm_B$?
  In DH-CCS, we just take $\rho^\pm_A := \rho^\pm|_{\{x\}}$ and 
  $\rho^\pm_B := \rho^\pm|_{\{y\}}$.
  We have no trouble inferring the global store $\rho^\pm$ on the 
  basis of these two fragments: 
  $\rho^\pm = \rho^\pm_A \cup \rho^\pm_B$.
  But if we were working in a semantics where $\rho$ is a density 
  matrix and we defined the fragments by partial traces, this 
  reconstruction of the global store would fail.
  The two partial traces would be 
  $\tr_x \rho^\pm = \tr_y \rho^\pm = \frac{1}{2}(\ket{0}\bra{0} + 
  \ket{1}\bra{1})$, which collapses the $\pm$ difference.
  By taking partial traces to localize the store, we have lost the 
  distinction between $\rho^+$ and $\rho^-$.
\end{example}

What else can we do in a semantics with density matrices as the 
state?
The SOS rules of existing calculi simply copy the global store to all 
subprocesses, shown here in DH-CCS notation:
\begin{mathpar}
  \inferrule[MonoCong]{\st{P}{\rho} \sstep{a} \st{P'}{\rho'}}%
  {\st{P \parcomp Q}{\rho} \sstep{a} \st{P' \parcomp Q}{\rho'}}
  \and
  \inferrule[MonoComm]{\st{P}{\rho} \sstep{c!x} \st{P'}{\rho'} \and 
    \st{Q}{\rho} \sstep{c?x} \st{Q'}{\rho'}}%
  {\st{P \parcomp Q}{\rho} \sstep{\tau} \st{P' \parcomp Q}{\rho'}}
\end{mathpar}
Notice how both the premises and the conclusion use the same $\rho$ 
before transition and $\rho'$ after.

This copy partition strategy amounts to giving local evolution a view 
of the whole system state.
Copy partition is good enough for reassembling the global state right 
after partitioning like in \cref{eg:partial-trace-problem}; however, 
it runs into trouble as soon as local processes begin evolving on 
their own.

\begin{example}\label{eg:copy-factorization-problem}
  This example shows copy partition also loses information.
  Consider
  \begin{mathpar}
    P^+ := \Bell^+(x,y). (Z(x).A(x) \parcomp I(y).B(y))%
    \and \text{and} \and%
    P^- := \Bell^-(x,y). (I(x).A(x) \parcomp Z(y).B(y))
  \end{mathpar}
  where $I(x)$ denotes the identity unitary transformation---a no-op 
  on $x$ formulated as a unitary to pad the transitions with a $\tau$ 
  to make them identical.
  The $A(x)$ and $B(x)$ are some named recursive processes, as 
  before.
  In $P^\pm$, the global state is initialized to $\Phi^\pm$ 
  respectively; consider what happens to these states in a density 
  matrix-based semantics where the process-local view is defined to 
  be just the global state.
  Noting that $Z\otimes I$ and $I\otimes Z$ flip $\Phi^+$ to $\Phi^-$ 
  and vice versa:
  \begin{itemize}
  \item In $P^+$, the $Z(x).A(x)$ flips $\Phi^+$ to $\Phi^-$ before 
    entering $A(x)$, while the $I(y).B(y)$ keeps it as $\Phi^+$ 
    before entering $B(y)$.
  \item In $P^-$, the $I(x).A(x)$ keeps $\Phi^-$ as $\Phi^-$ while 
    the $Z(y).B(y)$ evolves it to $\Phi^+$.
  \end{itemize}
  If we pause the execution when the process term has evolved to 
  $A(x) \parcomp B(y)$, what should the global store be?
  It could be $\Phi^+$ or $\Phi^-$---there's no way to tell based on 
  the local views, because in both $P^+$ and $P^-$, the $A(x)$ sees a 
  store evolved to $\Phi^-$ while $B(y)$ sees $\Phi^+$.
  Again, with DH-CCS there is no problem: the local views are just 
  the restrictions to $\{x\}$ and $\{y\}$, respectively, and however 
  they evolve, the resulting global state is their union.
\end{example}

In a monolithic semantics where the local view is not local at all 
but a copy of the global view, the only way not to lose sight of the 
correct global view appears to be to synchronize and recompose on 
every step.
This is exactly what the \textsc{MonoCong} rule above does: it 
restricts the transition of $P\parcomp Q$ to having only one of $P$ 
or $Q$ move, and the more up-to-date local view is immediately 
declared the global view.
This synchronization must be done at every step of each process.
Communication is the only case a monolithic semantics can allow both 
subprocesses to move simultaneously while ensuring 
local-global state correspondence is not lost.
This is because in a communication event, the global store does not 
change at all---it's only the ownership ledger that changes.
Thus, the communication message in a monolithic semantics transfers 
only the name of the qubit.

By contrast, DH-CCS has no trouble following the execution of each 
individual process for as long as we like.
No matter what they do, their local views are the respective 
restrictions of the global store.
At any moment, the corresponding global store is just the union of 
the processes' local-view stores.
This is the key advantage of DH descriptors, which enables modular 
reasoning and open system modeling.

So to repeat, what we mean by spatial compositionality is that every 
configuration $\st{P \parcomp Q}{\rho}$ splits into local 
configurations $\st{P}{\rho_P}$ and $\st{Q}{\rho_Q}$, that they can 
be evolved independently into $\st{P'}{\rho_P'}$ and 
$\st{Q'}{\rho_Q'}$, and that these results can be put back together 
to infer the evolved joint configuration $\st{P'\parcomp Q'}{\rho'}$.
This recomposition doesn't need any input besides the local 
configurations.

There are two complications to this understanding, which must be 
handled in order to formalize the notion of spatial compositionality:
\begin{enumerate}
\item \label{enum:consistency} We can only hope to piece together 
  local transition results that are consistent with each other.
  For example, if $x \in \FV P$, and if the execution of $Q$ ends up 
  with $x$ in its store, then the $P$ had better have sent that $x$ 
  out---that's the only way $Q$ can receive $x$ if $P$, $Q$ are 
  evolving together as $P \parcomp Q$.
  But taken out of context, $Q$ might well have received $x$ from 
  elsewhere, as $P$ may not even exist for all we know; hence, there 
  is nothing in the semantics to prevent a receive construct $c?x$ 
  within $Q$ from emitting an action $c?[x \mapsto \delta]$, a 
  request to receive qubit $x$, when $P$ doesn't contain $c!x$ 
  anywhere.
  Inconsistencies like this must be ruled out for recomposition to be 
  possible.
\item \label{enum:limited-info} We need to be careful what to give to 
  the recomposition procedure inferring $\rho'$ from $\rho_P'$ and 
  $\rho_Q'$.
  If it had access to entire local evolution histories, for example, 
  then recovering the global state would simply be a matter of 
  applying the SOS rules for parallel composition.
\end{enumerate}

For concern \cref{enum:consistency}, we restrict our attention to 
local transitions that result from partitioning a global transition.
Of course, we don't want to devolve into a tautology like ``if there 
is a global evolution consistent with the local evolution, then there 
is a global evolution.''
The point is that we can infer the global evolution from limited 
input about the local transitions, bringing us to concern 
\cref{enum:limited-info}.

For concern \cref{enum:limited-info}, we limit the recomposition 
process to working with just the final local evolution results 
$\st{P'}{\rho_P'}$ and $\st{Q'}{\rho_Q'}$.
These data suffice for recomposition in DH-CCS, but we will also show 
that adding local traces into the input set does not help with 
recomposition.
The point is that one needs an intensional view of what local 
processes have been up to in order to recompose---extensional 
information like traces are not enough.

\begin{definition}
  Denote trace concatenation by juxtaposition.
  When $S$ is a set of traces, $tS := \{ts \mid s \in S\}$.
  The set of all traces is written $\Trace$.
  The set of \emph{interactions} of a pair of traces $t,s$ is defined
  \begin{align*}
    \interact(\varepsilon, t) &:= \interact(t, \varepsilon) := \{t\} \\
    \interact((c!q)t,(c?q)s) &:= (c!q)\interact(t, (c?q)s)
                               \cup (c?q)\interact(t, (c!q)s)
                               \cup \tau \interact(t, s) \\
    \interact((c?q)t,(c!q)s) &:= (c!q)\interact((c?q)t, s)
                               \cup (c?q)\interact(t, (c!q)s)
                               \cup \tau \interact(t, s) \\
    \interact(at,bs) &:= a\interact(t, bs) \cup b\interact(at, s)
  \end{align*}
  where the first applicable rule wins.
\end{definition}

The interactions are just the set of traces that could be emitted 
from the parallel composition of processes emitting $t,s$ 
respectively.
It is the set of all interleavings of $t,s$ with zero or more 
occurrences of $(c!q)(c?q)$ or $(c?q)(c!q)$ collapsed to $\tau$, 
indicating an application of the \textsc{Comm} rule.

\begin{definition}[Spatial Compositionality]
  \label{thm:spatial-compositionality}
  A calculus is \emph{spatially compositional} iff it admits
  \begin{itemize}
  \item A \emph{partition function} 
    $\partition : \Proc^2 \rightarrow \State \rightarrow \State^2$,
  \item A partial \emph{recomposition function} 
    $\recons : (\Trace\times\Proc)^2 \rightarrow \State^2 
    \rightharpoonup \State$ \hspace{1em}(note the $\rightharpoonup$)
  \end{itemize}
  satisfying the following \emph{compositionality condition}: if 
  $\st{P\parcomp Q}{\rho} \sstep[*]{u} \st{P'\parcomp Q'}{\rho'}$, 
  then $\st{P}{\rho_P} \sstep[*]{t} \st{P'}{\rho_P'}$ and 
  $\st{Q}{\rho_Q} \sstep[*]{s} \st{Q'}{\rho_Q'}$ for some 
  $\rho_P', \rho_Q', t$ with 
  $\recons_{t,P',s,Q'} (\rho_P', \rho_Q') = \rho$ and 
  $u \in \interact(t,s)$.
\end{definition}

\noindent For DH-CCS, the partition and recomposition functions are of course 
restriction and union.

\begin{theorem}\label{thm:spatial-compositionality}
  DH-CCS as defined in \cref{sec:dh-ccs} is spatially compositional, 
  where:
  \begin{itemize}
  \item The partition function is 
    $\partition_{P,Q}(\rho) := (\rho|_{\FV P}, \rho|_{\FV Q})$.
  \item The recomposition function is 
    $\recons_{t,P,s,Q}(\rho_1, \rho_2) = \rho_1 \cup \rho_2$.
  \end{itemize}
  \begin{proof}
    Straightforward induction.
  \end{proof}
\end{theorem}

Let us call a \emph{state-only} semantics one where the $\rho$ of a 
configuration $\st{P}{\rho}$ carries the same information as the 
density matrix of the system's quantum state.
What \cref{eg:partial-trace-problem,eg:copy-factorization-problem} 
show is that a state-only semantics cannot be spatially compositional 
under the obvious partition schemes.

\begin{proposition}
  In a state-only semantics, neither the partial trace nor the copy 
  map $\rho \mapsto (\rho,\rho)$ work as a partition function.
  \begin{proof}
    The examples exhibit cases where two distinct global 
    configurations produce identical pairs of local configurations.
    Even the evolution traces are identical, so by the state-only 
    assumption, there is no other information to go on, and a 
    recomposition function cannot exist.
  \end{proof}
\end{proposition}

To the best of our knowledge, there is no other reasonable candidate 
for the partition function for density matrices in a state-only 
semantics.
To enable recomposition, it seems we need to carry more information 
than density matrices.
What can that information be?

An answer can be gleaned from 
\cref{eg:partial-trace-problem,eg:copy-factorization-problem}.
Notice how, in each case, we could successfully recover the global 
configuration from local ones if only we knew which global 
configuration we started from, or what operations the local processes 
performed during their evolution.
In other words, if we know something about the history of the 
evolution, we could recompose.

This leads us to the idea of a \emph{historical semantics}, a simple 
formulation of which goes as follows.
The $\rho$ is a mapping from variables to lists of unitaries 
performed on each qubit, e.g.
\[
  [x \mapsto [X(x), \CX(x \rightarrow y), \ldots], y \mapsto [\CX(x \rightarrow y), Z(y), \ldots]]
\]
where each variable lists only the unitaries that involved it, in 
chronological order.
The transition rules are identical to the standard semantics, except 
that unitary application simply records the unitary instead of 
performing any semantic action:
\begin{equation*}
  \inferrule{ }{\st{U(x_1,\ldots,x_n).P}{\rho} \sstep{\tau} 
    \st{P}{\rho[\forall j.\ x_j \mapsto \rho(x_j) \concat [U(x_1, \ldots, x_n)]]}}
\end{equation*}
where $\concat$ denotes concatenation.

\begin{proposition}
  Historical semantics is spatially compositional, with the partition 
  function, recomposition function, and consistency predicate defined 
  analogously to the standard DH-CCS semantics.
\end{proposition}

Historical semantics is a sort of syntactic model for distributed 
quantum computation, and it is just as unsatisfactory as syntactic 
denotational models generally are.
A concrete problem with historical semantics is that it is not 
\emph{state-like}: it encodes not the current state of a system but 
lugs around its entire evolution history, so it's not informative 
about what quantum information looks like or where it is localized 
at.
This unstate-like nature is evidenced by how the size of $\rho$ grows 
unboundedly over time, even when the number of qubits at play is 
bounded.

DH descriptors, by contrast, have storage size bounded by a function 
in the number of qubits involved.
As discussed in \cref{sec:dh,sec:splittable-dh}, DH descriptors are 
multilinear expressions in the Pauli symbols, by which we mean the 
$X_x$'s and $Z_x$'s.
If the number of qubits involved is fixed, then the number of 
distinct Pauli symbols is fixed, and there are only a fixed number of 
distinct monic, multilinear monomials in those symbols.
The entire multilinear expression can be seen as a list of 
coefficients for those monomials, so the size of the whole store is 
bounded.

Thus, DH descriptors can be seen as a way to compress the history by 
exploiting the algebraic properties of unitary operators.
It is no coincidence that our construction of nominal DH descriptors 
had a ``syntax modulo equalities'' vibe: it's an encoding of history, 
which is syntactic, but compressed using semantic equalities.

\begin{remark*}
  Lurking in this discussion is, of course, an annoying question 
  about how to measure the size of numbers.
  We are counting complex numbers as unit size here, but 
  infinite-precision complex numbers offer infinite information 
  storage.
  It is possible to encode unbounded history in a single complex 
  number, say in the binary expansion of its real part, so if we 
  count any complex number as having unit size, a variation on 
  historical semantics guaranteeing unit-size storage would exist.
  We can try to patch this loophole by carefully defining what usage 
  of numbers count as unit size (e.g.\ by defining the set of complex 
  numbers that may arise in states reachable in $k$ steps), but the 
  intellectual payout of this exercise is minuscule.
  We simply trust the reader to understand what we mean instead.
\end{remark*}

\section{Process Equivalence}
\label{sec:bisimulation}

In this section, we define an observational equivalence for DH-CCS to 
clarify which aspects of a process' behavior is considered to matter.
The equivalence can be seen as a barbed bisimulation 
\cite{MilnerSangiorgi1992BarbedBisimulation}, though with a somewhat 
unusual barb that observes values instead of actions.
It quantifies over all observers, so in that sense, it is also a form 
of contextual equivalence.

Let us start with the formal definitions and defer rationales.
There are enough aspects to discuss that it will help to have the 
formal definition upfront as a point of reference.
An important auxiliary definition is that of name conversions, used 
to collapse differences in the way qubits are named---it's only the 
relative roles of the qubits that matter, not their names.

\begin{definition}
  A \emph{name conversion} is an injection 
  $\gamma : V \hookrightarrow \Var$ for $V \subseteq_\fin \Var$.
  Two name conversions $\gamma_1, \gamma_2$ are \emph{consistent} iff 
  $\forall x \in \dom \gamma_1 \cap \dom \gamma_2.\ \gamma_1 (x) = 
  \gamma_2 (x)$.
  Define the action of $\gamma$ on $\rho$ by 
  $\gamma \rho := \gamma\circ \rho\circ \gamma^{-1}$ using the 
  partial inverse $\gamma^{-1}$.
\end{definition}

\noindent Note that the action of $\gamma$ on $\rho$ is just the same 
action when $\gamma$ is extended to a permutation on $\Var$, provided 
$\dom \gamma \supseteq \supp \rho$.

\begin{definition}\label{defn:valid-rho}
  A store $\rho$ is \emph{valid} iff it satisfies the Pauli 
  equations.
\end{definition}

\begin{definition}\label{defn:bisim}
  Let $\st{P}{\rho} \Sstep{a} \st{P'}{\rho'}$ denote weak transition, 
  i.e.\ an $a$-transition preceded and followed by zero or more 
  $\tau$-transitions.
  An equivalence relation $R_\Gamma$ indexed by a set of name 
  conversions $\Gamma$ is an \emph{observer bisimulation} iff 
  $\st{P_1, O_1}{\rho_1} \mathbin{R}_\Gamma \st{P_2, O_2}{\rho_2}$ implies 
  that $\st{P_1\parcomp O_1}{\rho_1}$ and 
  $\st{P_2\parcomp O_2}{\rho_2}$ are well-formed configurations, and:
  \begin{enumerate}
  \item \label{enum:bisim:empty} $\Gamma \neq \varnothing$.
  \item \label{enum:bisim:coverage}
    $\forall \gamma \in \Gamma.\ \dom \gamma = \supp \rho_1|_{\FV O_1} \land 
    \img \gamma = \supp \rho_2|_{\FV O_2}$.
  \item 
    $\forall \gamma \in \Gamma.\ \density (\gamma (\rho_1|_{\FV O_1})) = 
    \density (\rho_2|_{\FV O_2})$.
  \item \label{enum:bisim:valid} $\rho_1$ and $\rho_2$ are valid.
  \item \label{enum:bisim:follows} If 
    $\st{P_1 \parcomp O_1}{\rho_1} \Sstep{a_1} 
    \st{P_1' \parcomp O_1'}{\rho_1'}$ then 
    $\st{P_2 \parcomp O_2}{\rho_2} \Sstep{a_2} 
    \st{P_2' \parcomp O_2'}{\rho_2'}$ such that 
    $\st{P_1',O_1'}{\rho_1'} \mathbin{R}_{\Gamma'} 
    \st{P_2',O_2'}{\rho_2'}$ for some $\Gamma'$ where every 
    $\gamma' \in \Gamma'$ is consistent with at least one 
    $\gamma \in \Gamma$.
  \end{enumerate}
  Observer bisimilarity, written $(\obssim{\Gamma})$, is the largest 
  observer bisimulation.
  \emph{Observer equivalence} $P \obseq Q$ holds when 
  $\st{P, O}{\rho|_{\FV (P \parcomp O)}} \obssim{\Gamma_O} \st{Q, 
    O}{\rho|_{\FV (Q \parcomp O)}}$ for any $O$ and $\rho$ satisfying 
  well-formedness conditions, where $\Gamma_O := \{\id_{\FV O}\}$.
  We call the $O$ in $\st{P,O}{\rho}$ an \emph{observer}, the $P$ a 
  \emph{subject}, and the configuration 
  $\st{(P \parcomp O) \setminus \Chan}{\rho}$ and its trajectory of 
  reductions an \emph{experiment}.
\end{definition}

Now let us argue that this definition faithfully captures the notion 
of distinguishability by physically allowable processes.
The starting point for this definition is that DH descriptors are 
\emph{gauge}, making finer distinctions than density matrices.
This is considered a problem because density matrices are the gold 
standard of a quantum system's information content, and any finer 
distinctions deserve to be quotiented away 
\cite{WallaceTimpson2007NonlocalityGaugeFreedom}---but do they all?

Gauge freedom should be familiar to many readers who have worked with 
ensemble-based semantics (cf.~\cref{sec:unification}).
Discarding the first qubit in the Bell state 
$\Phi^+ = \frac{1}{\sqrt{2}}(\ket{00} + \ket{11})$ leaves 
$(\frac{1}{2}\ket{0}, \frac{1}{2}\ket{1})$, but an Hadamard on the 
first qubit before discarding changes that to 
$(\frac{1}{2}\ket{+}, \frac{1}{2}\ket{-})$.
Both are the maximally mixed state, described by the same density 
matrix, and indeed, these states are experimentally 
indistinguishable.
The work on quantum process equivalence by Ceragioli et al.\ 
\cite{CeragioliEtAl2024QuantumBisimilarityBarbs} devotes substantial 
effort to preventing observers from exploiting the distinction 
between these ensembles in order to ensure a physically justified 
equivalence.

But DH descriptors' gauge freedom is partly of a different nature.
For instance, the maximally mixed state on a local qubit $x$ could 
arise due to entanglement with a remote qubit $y$ or with remote 
qubit $z$, and DH descriptors remember which.
Both states give the same density matrix; yet, for purposes of future 
communication, it very much matters which one is at hand.
Density matrices collapse this difference, forgetting details of how 
the system was entangled with the outside world.

DH descriptors also have their share of distinctions that are simply 
unwanted.
For example, preparing the Bell state by Hadamard + CX in the usual 
way gives two slightly different descriptors based on which of the 
two qubits was fed to the Hadamard 
\cite{HorsmanVedral2007DevelopingDeutschHayden}.
The gauge freedom of ensembles come from over-specifying the 
decomposition of a state into smaller states; the gauge freedom of DH 
descriptors come from over-specifying the state's preparation 
history.

We don't want to lose the distinctions made on entanglement 
structure, and these are distinctions that only matter in the future, 
once further communication happens.
As such, we want to compare processes by testing for current 
equivalence by collapsing the quantum states to density matrices, 
while also looking out for future divergence.
Recalling that density matrices are a summary of measurement 
probabilities (see discussions leading up to \cref{eq:dh-density}), 
this amounts to equivalence under an observer that can terminate the 
experiment at any time and measure the qubits.
Note that this does not mean the observer actually performs those 
measurements constantly.
The observer only \emph{threatens} to do so, saying that putatively 
equivalent processes had better not create a window of opening for 
distinguishing measurements, or else.

This current-density comparison is a continuous monitoring process, 
so the comparison takes the form of a bisimulation, with the 
current-density comparison as a barb.
The observer is initially identical on both sides, so they operate 
according to the same blueprint, but they may drift apart due to 
differences in how the subjects behave, provided those differences do 
not trigger the barb.

Moreover, the current-density check should be restricted to qubits 
accessible to the observer, to avoid observing internal bookkeeping 
data.
For example, a process that computes with a single qubit should be 
identified with another process performing the same algorithm on a 
triplicated qubit; yet, if we gave the observer access to internal 
qubits, not even the sizes of the density matrices would match!
The simplest way to account for which qubits are held by the observer 
is to make the observer an in-calculus process.
Thus, we get the form of experiment $\st{P\parcomp O}{\rho}$.

We must then guard the experiment against unphysical behavior by the 
environment.
An especially egregious violation of physics is what we call 
unilateral disentanglement: suppose the experiment sends out a qubit 
$x$ entangled with qubit $y$, but $y$ remains with the experiment.
There is nothing in the semantics thereafter to prevent the 
experiment from receiving the message $c?[x \mapsto (X_x, Z_x)]$, 
which claims that $x$ is now $\ket{0}$, disentangled from $y$, when 
the local DH descriptor for $y$ mentions $x$, proving entanglement.

There are two avenues for unphysical states to creep in, with two 
solutions.
One is incoming messages $c?q$ as we just mentioned.
The simple solution here would be to close off the experiment like 
$(P \parcomp Q) \setminus \Chan$.
This is actually not a bad choice for observer bisimulation, as it 
helps to curb the number of cases one has to consider when reasoning 
directly with this equivalence.
No generality is lost: if there are any external processes the 
subject might want to communicate with, those can be included in the 
observer.

Unfortunately, closed experiments don't help to plug the other source 
of unphysical states: the store $\rho$ with which the experiment 
begins.
To rule out unphysical stores, we need to characterize physically 
realizable stores.
Fortunately, we have just such a characterization---the Pauli 
equations from \cref{defn:su2-algebra} (adapted to nominal DH 
descriptors in the obvious way).
We've noted that the physically realizable DH descriptor sets are 
those of the form $U^\dagger \rho_0(V) U$, and that they satisfy the 
Pauli equations.
Perhaps surprisingly, the converse also holds, and to boot, for 
partial descriptor sets that don't assign descriptors to all 
variables under consideration.

\begin{lemma}\label{thm:sqrt-polynomial}
  If $P$ is a positive matrix, then $\sqrt{P}$ is a polynomial in $P$.
  \begin{proof}
    By the spectral theorem, 
    $P = \diag(\lambda_1, \ldots, \lambda_n)$ and 
    $\sqrt{P} = \diag(\sqrt{\lambda_i}, \ldots, \sqrt{\lambda_n})$ in 
    some ONB.
    We can fit a polynomial $p$ to the points 
    $(\lambda_i, \sqrt{\lambda_i})$ so that 
    $\forall i.\ p(\lambda_i) = \sqrt{\lambda_i}$.
    Then $p(P) = \sqrt{P}$.
  \end{proof}
\end{lemma}

\begin{theorem}[Descriptor Completion]\label{thm:completion}
  Let $\rho : V \rightarrow \Obsvs$ and $W := \supp \rho$.
  If $\rho$ satisfies the Pauli equations, then there exists a 
  completion $\mu : W \rightarrow \Obsvs$ such that $\mu|_V = \rho$ 
  and $\mu = U^\dagger \rho_0(W)U$ for some unitary $U$.
  \begin{proof}
    It suffices to prove this for $\rho : V \rightarrow \Obsvs_W$; 
    the results then transfer to $\Obsvs$ through the embedding 
    $\iota_W$.
    Let us write $X_x^V, Z_x^V$ to mean 
    $X_x, Z_x \in \mathbb{C}^{2^V\times 2^V}$ and $X_x^W, Z_x^W$ to 
    mean $X_x, Z_x \in \mathbb{C}^{2^W\times 2^W}$.
    One can check that the Pauli equations ensure that the mapping 
    from single-qubit Paulis
    \begin{mathpar}
      \forall x \in V. \and \varphi(X_x^V) := \rho(x)_X \and \varphi(Z_x^V) 
      := \rho(x)_Z
    \end{mathpar}
    extends to a unital $\mathbb{C}$-algebra homomorphism 
    $\varphi : \mathbb{C}^{2^V\times 2^V} \rightarrow 
    \mathbb{C}^{2^W\times 2^W}$ preserving ${}^\dagger$.

    By the Skolem-Noether theorem 
    \cite[Thm.~4.9]{Jacobson2009BasicAlgebraII}, $\varphi$ is 
    conjugate to 
    $-\otimes 1 : \mathbb{C}^{2^V\times 2^V} \hookrightarrow 
    \mathbb{C}^{2^W\times 2^W}$, i.e.\ 
    $\varphi(A) = S^{-1}(A\otimes 1)S$ for some invertible 
    $S \in \mathbb{C}^{2^W\times 2^W}$.
    By the polar decomposition, $S = \sqrt{SS^\dagger}U$ for a 
    unitary $U$.
    Expanding out $\varphi(A^\dagger) = \varphi(A)^\dagger$ and 
    rearranging gives 
    $(A\otimes 1)SS^\dagger = SS^\dagger (A\otimes 1)$.
    By \cref{thm:sqrt-polynomial}, $\sqrt{SS^\dagger}$ must also 
    commute with $A\otimes 1$, and 
    $\varphi(A) = U^\dagger\sqrt{SS^\dagger}^{-1}(A\otimes 
    1)\sqrt{SS^\dagger}U = U^\dagger(A\otimes 1)U$.

    Thus, $\rho(x)_X = \varphi(X_x^V) = U^\dagger X_x^W U$ and 
    $\rho(x)_Z = \varphi(Z_x^V) = U^\dagger Z_x^W U$ for $x \in V$.
    Defining $\mu|_V := \rho$ and $\mu(y)_X := U^\dagger X_y^W U$ and 
    $\mu(y)_Z := U^\dagger Z_y^W U$ for $y \in W-V$ gives the 
    promised completion.
  \end{proof}
\end{theorem}

Therefore, a partial store, regardless of the source, is physically 
realizable iff it satisfies the Pauli equations.
Thus, checking for these equations amounts to checking for physical 
realizability, both for communicated qubits and qubits found in the 
store.
In \cref{defn:bisim}, we check for store validity in clause 
\cref{enum:bisim:valid}, at each step of the observer bisimulation.

Current-density comparisons must be performed modulo renaming of 
qubits because names are not supposed to matter 
(cf.~\cref{thm:equivariance}).
When the observer receives $c?[x \mapsto \delta_1]$ from $P_1$ but 
$c?[y \mapsto \delta_2]$ from $P_2$, that can simply signify a 
different memory management technique having no bearing on the 
computational content.
Thus, upon receipt of a qubit, the observer must act as follows, 
depending on whether the $x,y$ pair is novel.
\begin{itemize}
\item If the pair is new to the observer, this name difference is not 
  grounds for distinction.
  The observer must carry on, noting the correspondence 
  $[x \mapsto y]$ that accumulates into a name conversion $\gamma$.
  If there is more than a naming mismatch, it will be detected in the 
  current-density comparison that comes immediately after.
\item If the pair is known to the observer (having seen qubits 
  entangled with $x,y$ before but not $x,y$ themselves), there is an 
  entry in the conversion table for $x$, so $\gamma x = y$ should be 
  checked; failure is clear grounds for distinction of the subjects.
  However, this check is subsumed by the current-density check, so 
  all the observer needs to do is carry on.
\end{itemize}

In fact, in both cases, $\gamma$ can't be updated properly until the 
current-density comparison.
The $\delta_1$ and $\delta_2$ may contain further novel variables as 
representations of entanglement, and these cannot simply be compared 
textually due to the gauge freedom of DH descriptors.
The only way to compare them, entanglement and all, is to incorporate 
them into the observers' stores and compare them under extensions 
$\gamma'$ of $\gamma$.

When current density matrices do turn out to be equivalent, there is 
not always a unique $\gamma'$ that successfully identify them, so we 
do not commit to any one of them, keeping around the whole set 
$\Gamma'$.
But any new $\gamma' \in \Gamma'$ does have to be an extension of the 
old $\gamma$, so $\Gamma'$ is restricted to those $\gamma'$ which are 
consistent with $\gamma$.
The comparison fails when no viable $\gamma'$ exists (clause 
\cref{enum:bisim:empty} of \cref{defn:bisim}).

Punting the handling of name conversions to the density comparison 
this way has the advantage that allocation is automatically taken 
care of as well.
Unlike the $c?q$ construct, allocation does not produce a visible 
action.
Instead of looking at actions, the actual equivalence looks at the 
store after the qubit is incorporated and checks that $\gamma'$ 
covers the extended store on both sides (clause 
\cref{enum:bisim:coverage}).

Sending and discarding works similarly, except they shrink the name 
conversion.
When the observer emits $c![x \mapsto \delta]$ or discards 
$[x \mapsto \delta]$, the store shrinks, and $\Gamma$ is updated 
accordingly by dropping $x$ and the variables in $\delta$, except for 
those that appear in the remaining store (indicating lingering 
entanglement).
Once again, the name conversion should remain unchanged for those 
names that remain visible to the observer, so new name conversions 
are checked to be consistent with the old one (last part of clause 
\cref{enum:bisim:follows}).

As a final point of caution, we must \emph{not} observe the $\delta$ 
part of transition labels $c![x \mapsto \delta]$ and 
$c?[x \mapsto \delta]$ coming out of the subject.
Gauge freedom means that observations of these descriptors must be 
left to the current-density comparison.
That leaves us with observation of the communication outline, by 
which we mean the $c!$ or $c?$ parts.
Observing outlines is superfluous because the observer can see them.
It can have qubits $x_{c!}$ and $x_{c?}$ for each channel $c$ it 
touches, and instead of doing (say) $c?x.P$, it can do 
$c?x.X(x_{c?}).P$, flipping $x_{c?}$ as soon as a $c?$ message 
arrives.
Then any differences in communication outlines is immediately flagged 
to the current-density comparison.
Therefore, observer bisimulations need not check labels explicitly at 
all.

\ifproofdetails


  \begin{lemma}\label{thm:fresh-equivariant}
    If $x$ is fresh for some set of names $V$, then $\sigma x$ is 
    fresh for $\sigma V$.
    \begin{proof}
      If $\sigma x \in \sigma V$, then by definition of $\sigma V$, 
      we have $\sigma x = \sigma y$, hence $x = y$, for some $y \in V$.
    \end{proof}
  \end{lemma}

  \begin{lemma}
    For vectors $u \in \mathbb{C}^{(2^V)}$ and 
    $v \in \mathbb{C}^{(2^W)}$ where $V, W$ may be infinite, 
    $\sigma (u\otimes v) = (\sigma u)\otimes (\sigma v)$.
    \begin{proof}
      $\sigma (u\otimes v) (\eta \cup \bar{\eta}) = (u\otimes 
      v)(\sigma^{-1} \eta \cup \sigma^{-1} \bar{\eta}) = 
      (u(\sigma^{-1} \eta)) \otimes (v(\sigma^{-1} \bar{\eta})) = 
      ((\sigma u) \eta) \otimes ((\sigma v) \bar{\eta}) = ((\sigma u) 
      \otimes (\sigma v)) (\eta \cup \bar{\eta})$.
      Bear in mind these are vectors, i.e.\ maps from valuations to 
      complex numbers, so there's no outer $\sigma$.
    \end{proof}
  \end{lemma}

  \begin{lemma}
    For  linear endomorphisms on different spaces,
    $\sigma (A \otimes B) = \sigma A \otimes \sigma B$.
    \begin{proof}
      Apply both sides to vectors and it reduces to the case for 
      vectors: 
      $\sigma (A\otimes B) (u\otimes v) = \sigma ((A\otimes B) 
      (\sigma^{-1} (u\otimes v))) = \sigma (A (\sigma^{-1}u) \otimes 
      B (\sigma^{-1}v)) = (\sigma A u)\otimes (\sigma B v) = ((\sigma 
      A) \otimes (\sigma B)) (u\otimes v)$.
    \end{proof}
  \end{lemma}

  \begin{lemma}
    $\delta : \mathbb{C}^{(2^\Var)} \rightarrow \mathbb{C}^{(2^\Var)}$ is 
    in $\Domain$ iff $\supp \delta$ is finite.
    \begin{proof}
      \fixme{See my notes on this paper.
        The case where $\#(\Var - V) = 2$ is solved.}
    \end{proof}
  \end{lemma}

  \begin{lemma}\label{thm:domain-equivariant}
    If $\delta \in \Domain$, then $\sigma \delta \in \Domain$.
    \begin{proof}
      $\supp (\sigma \delta) = \sigma (\supp \delta)$, which 
      preserves finite-ness.
    \end{proof}
  \end{lemma}
\fi

Now, let us turn to analyzing the properties of $(\obseq)$.
It is evidently an equivalence.
Being based on a weak bisimulation, it is not a congruence for the 
usual reason: it holds that $\tau.c!x.0 \obseq c!x.0$, yet 
$\tau.c!x.0 + d!x.0$ can be distinguished from $c!x.0 + d!x.0$ by the 
observer $c?x.!!x.0 + d?x.!!x.0$ over any store.


Observer equivalence is like a contextual equivalence in that it 
requires reasoning over all possible observers.
As usual, we define a bisimulation to help establish it without 
having to look at all contexts.

The idea is to exploit spatial compositionality and reason about just 
the subjects.
In observer equivalence, we had to guard against qubit naming 
inconsistencies in the subjects; however, now that we have access to 
the subjects, we can simply ask them to use equivariance 
(\cref{thm:equivariance}) to align their naming schemes before 
comparison, at least for the outbound qubits.
Thus, name conversions can be dropped without loss of generality.

\begin{definition}
  A (weak) \emph{subject bisimulation} is an equivalence relation $R$ 
  on well-formed configurations such that 
  $\st{P_1}{\rho_1} \mathbin{R} \st{P_2}{\rho_2}$ implies that:
  \begin{itemize}
  \item $\rho_1$ and $\rho_2$ are valid.
  \item If $\st{P_1}{\rho_1} \Sstep{a} \st{P_1'}{\rho_1'}$ then 
    $\st{P_2}{\rho_2} \Sstep{a} \st{P_2'}{\rho_2'}$ such that 
    $\st{P_1'}{\rho_1'} \mathbin{R} \st{P_2'}{\rho_2'}$.
  \end{itemize}
  Subject bisimilarity $(\subsim)$ is the largest subject 
  bisimulation.
  Process terms are \emph{subject equivalent}, written 
  $P_1 \subeq P_2$ iff 
  $\st{P_1}{\rho|_{\FV P_1}} \subsim \st{P_2}{\rho|_{\FV P_2}}$ for 
  any valid $\rho$ subject to well-formedness conditions.
\end{definition}

Subject bisimulation requires that action labels $a$ be matched 
exactly.
This is partly because we can assume names have been normalized, but 
more importantly because we want to avoid estimating the state of the 
observer.
If we want to collapse the gauge freedom of DH descriptors, we would 
track the qubits that the subject has sent out and do a 
current-density comparison.
But when a qubit comes back from the observer, \cref{thm:completion} 
does not give a unique completion, so we would have to perform 
current-density comparisons for all possible completions.
While this is probably possible to pull off, we believe the above 
formulation strikes a better balance between simplicity and 
expressivity.

\begin{proposition}
  If $P_1 \subeq P_2$ then $P_1 \obseq P_2$.
  \begin{proof}
    Given $P_1 \subeq P_2$, define a family of equivalences 
    $\sim_\Gamma$ by 
    $\st{Q_1,O_1}{\rho_1} \subsim_{\Gamma} \st{Q_2,O_2}{\rho_2}$ iff:
    \begin{enumerate}
    \item \label{enum:subsim:same-observer} $O_1 = O_2$ and 
      $\Gamma = \{\id \text{ on } \supp (\rho_1|_{\FV O_1})\}$.
    \item \label{enum:subsim:same-store} 
      $\rho_1 = \rho|_{\FV (Q_1\parcomp O_1)}$ and 
      $\rho_2 = \rho|_{\FV (Q_2\parcomp O_2)}$ for some $\rho$.
    \item \label{enum:subsim:transfer}
      $\st{Q_1}{\rho_1|_{\FV Q_1}} \subsim \st{Q_2}{\rho_2|_{\FV 
          Q_2}}$.
    \end{enumerate}
    We claim $\sim_\Gamma$ is an observer bisimulation.
    Assume $\st{Q_1,O}{\rho_1} \sim_\Gamma \st{Q_2,O}{\rho_2}$.
    Clearly the observer bisimulation conditions on $\Gamma$ are 
    satisfied, while the current-density check passes because 
    $\rho_1|_{\FV O} = \rho|_{\FV O} = \rho_2|_{\FV O}$ and the only 
    $\gamma \in \Gamma$ is the identity.
    The interesting part is the transfer property.
    It should be clear that every move on one side has a matching 
    move on the other.
    The question is whether that move preserves clauses 
    \cref{enum:subsim:same-observer,enum:subsim:transfer} of the 
    definition of $\sim_\Gamma$, justifying coinduction.
    \begin{itemize}
    \item If the observer moves unilaterally, then by spatial 
      compositionality, the subjects' part of the store is unchanged, 
      so the other side can match the move exactly by having the 
      observer (which identical on both sides) do the exact same 
      move.
      Then clause \cref{enum:subsim:transfer} is unaffected.
    \item If the subject moves unilaterally, the resulting subject 
      configurations remain subject-bisimilar, upholding clause 
      \cref{enum:subsim:transfer}.
      Spatial compositionality ensures the observer is unaffected, so 
      clauses \cref{enum:subsim:same-observer,enum:subsim:same-store} 
      are also preserved.
    \item If there is communication, a qubit moves from the 
      observer's part of the store to the subjects' or vice versa.
      But both are fragments of the same valid store, so they both 
      remain valid.
      The observers do not diverge in form because transition labels 
      of the subjects (especially the $x$ part of 
      $c![x \mapsto \delta]$) are matched exactly.
      Clauses \cref{enum:subsim:same-observer,enum:subsim:transfer} 
      are therefore preserved.
      \qedhere
    \end{itemize}
  \end{proof}
\end{proposition}

\noindent We now show that observer equivalence behaves as expected 
on simple examples.

\begin{example}
  Assume a recursive process defined as $\sink(x) := \sink(x)$, i.e.\ 
  an infinite loop that clings onto a qubit.
  \begin{itemize}
  \item $c!x.0 \not\obseq c!y.0$ because they are distinguished by 
    the observer $O := c?z. \sink(z)$ with initial store 
    $\rho := [x \mapsto (X_x, Z_x), y \mapsto (X_y, -Z_y)]$ 
    (corresponding to $x := \ket{0}$, $y := \ket{1}$).
    The observer can terminate the experiment any time after $c?z$ 
    receives the $x$ or $y$ qubit, upon which measurement will reveal 
    which one came through.
    Thus, $(\obseq)$ can distinguish qubits with different names.
  \item $!!x.0 \obseq !!y.0 \obseq \sink(x)$ because they are subject 
    equivalent.
    Thus, qubits which stay internal to a process is immaterial.
    \qedhere
  \end{itemize}
\end{example}

\begin{example}
  This example shows that DH descriptors do not suffer some of the 
  gauge freedom of ensembles---they are simply different types of 
  freedom.
  Let $Mx.P$ be an abbreviation for $??t.\CX(x \rightarrow t).!!t.P$, 
  i.e.\ measurement of $x$ retaining the measurement result.
  Define:
  \begin{equation*}
    \begin{array}{@{}l@{}l@{}}
      P_1 := H(x).\CX(x \rightarrow y).     &!!x.A(y) + B(y) \\
      P_2 := H(x).\CX(x \rightarrow y).H(x).&!!x.A(y) + B(y)
    \end{array}
    \hspace{1em}\text{where}\hspace{1em}
    \begin{array}{@{}l@{}l@{}}
      A(y) := &My.c!y.0 \\
      B(y) := H(y).&My.d!y.0
    \end{array}
  \end{equation*}
  Both $P_1$ and $P_2$ prepare a Bell state $\Phi^+$ and discard one 
  qubit, but $P_2$ performs Hadamard on it just before discarding.
  Then $A(y)$ measures $y$ in the Z axis while $B(y)$ measures in the 
  X axis, sending out the result over a different channel to indicate 
  the choice that was taken.
  These terms are subject bisimilar because whatever is done to qubit 
  $x$ before discarding has no bearing on the state of the remaining 
  qubit $y$, hence of anything that happens thereafter.

  The $P_1$ and $P_2$ in this example are physically 
  indistinguishable because they both set $y$ to the maximally mixed 
  state before entering $A(y) + B(y)$, though prepared as a different 
  ensemble.
  But the transfer property of bisimulation would force $P_2$ to 
  follow the scheduling choice of $A(y)$ vs.\ $B(y)$ in $P_1$.
  This observation of scheduling decisions could be leveraged to 
  distinguish $P_1$ from $P_2$, which motivated the observer 
  restrictions in \cite{CeragioliEtAl2024QuantumBisimilarityBarbs}.
  DH-CCS is insulated from this problem, although it might not 
  survive introduction of classical control and ensembles: DH 
  descriptors are modular, so they naturally model the irrelevance of 
  the Hadamard on the outgoing $x$ qubit.
\end{example}




\section{Case Study: BB84 Quantum Key Distribution}
\label{sec:example}

In this section, we demonstrate reasoning on an open system, where a 
qubit leaves, interacts with an external process, and re-enters.
We use a fragment of BB84 
\cite{BennettBrassard2014QuantumCryptographyPublic} as an example, 
but the point is not to verify the protocol.
It is to show the strength of spatial compositionality for open 
system modeling, while also highlighting a current limitation: 
descriptor inspection alone does not suffice for general 
information-flow analyses.

Let us briefly review the protocol.
Alice and Bob want to share a secret random bit string.
They share a public classical communication channel $c_c$, 
authenticated to reliably reach each other, and a public, untrusted 
quantum channel $c_q$.
A single round of the protocol goes as follows.
\vspace{0.4\baselineskip}
\begin{paracol}{2}
  \noindent \underline{Alice's moves}:
  \begin{enumeratealigned}
  \item Flip a classical coin $a$ to select a basis.
  \item Flip a classical coin $k$ to decide on the secret key bit.
  \item Encode $k$ in a qubit $k'$ by the basis she chose; send $k'$ 
    over $c_q$.
  \item Receive $b$ over the classical channel $c$.
  \item If $a \neq b$, abort.
    If not, randomly decide: accept $k$ as shared bit, or use $k$ to 
    detect eavesdropping.
    Synchronize with Bob on decision.
  \item If the decision was detection, send $k$ over $c$.
  \end{enumeratealigned}
  \switchcolumn
  \noindent \underline{Bob's moves}:
  \begin{enumeratealigned}
  \item Flip a classical coin $b$ to select a basis.
  \item Receive $\ell$ over the quantum channel $c_q$ and measure by 
    his selected basis.
  \item Send $b$ over the classical channel $c$.
  \item Synchronize with Alice on the abort-accept-detect decision.
  \item If the decision was detection, receive $k$ over $c$.
  \item Test whether $\ell = k$.
    If $\ell \neq k$, declare an eavesdropper is present.
  \end{enumeratealigned}
\end{paracol}

The communication through $c_q$ can be arbitrarily affected by an 
attacker Eve, so $\ell$ might not be $k'$.
Randomized detection forces Eve to ensure $\ell$ always measures the 
same as $k'$ to cover her tracks.
Randomized choice of measurement basis forces her to preserve the 
eigenstates of both Z-measurements and X-measurements, preserving the 
qubit entirely.
That requires remaining unentangled with it, so Eve's qubits cannot 
be correlated with $\ell$.
One can establish from this line of argumentation that, subject to 
$a = b$, one of the following holds:
\begin{itemize}
\item Detection: There is non-zero probability that Bob sees 
  $\ell \neq k'$ if they try detection, or
\item Security: No eavesdropper can produce a bit correlated with 
  $k$.
\end{itemize}
This is the simplest version of the security property of BB84, 
usually coupled with a correctness statement that if $a = b$ and no 
eavesdropper exists, $\ell = k$ with certainty.

We model this protocol in DH-CCS with a few simplifications to focus 
on the parts relevant to spatial compositionality, without ruining 
the overall shape of the protocol:
\begin{itemize}
\item We encode everything in qubits.
  Modeling the protocol faithfully would require constantly measuring 
  encoded classical bits, but we omit that.
\item We analyze just the cases where $a = b$ happened to hold.
  Concretely, we allow Alice and Bob to exchange a Bell pair at the 
  beginning of the protocol.
\item To faithfully model the public nature of $c_c$, Alice and Bob 
  should CC an outbound channel whenever they communicate over $c_c$, 
  but this is also omitted.
\item For the sake of presentation, it helps to talk about the set of 
  all qubits Alice and Bob touched, so we will have them not discard 
  any qubits.
\item In a real implementation, Alice and Bob will go on to do 
  something with the established key bit or the detection verdict.
  This continuation will be modeled by processes $K_A$, $K_A'$, 
  $K_B$, and $K_B'$, but for simplicity, we will pretend as if these 
  terms have no outgoing transitions.
\end{itemize}

With those simplifications, the protocol can be modeled as follows.
We write prefixes like $c!x.P$ as $[c!x]P$ for readability.
We write $[??x=y]P$ for copying allocation, i.e.\ 
$[??x][\CX(y \rightarrow x)]P$.
\begin{gather*}
  \begin{array}[t]{@{}l@{\hspace{1em}}}
    c_b : \text{quantum channel for exchanging Bell pair} \\
    c_c : \text{general-purpose classical channel} \\
    c_d : \text{channel for synchronizing on ``detect'' decision} \\
    c_q : \text{untrusted quantum channel} \\
    d : \text{dummy payload}
  \end{array}
  \begin{array}[t]{l@{}}
    a : \text{Alice's basis} \\
    b : \text{Bob's basis} \\
    k : \text{Alice's secret bit, sent out as }k' \\
    \ell : \text{Bob's tampered secret bit}  \\
    v : \text{comparison verdict }\ell \neq k
  \end{array}
  \\
  \begin{aligned}
    A &:= [H(a)][??b=a][c_b!b][H(k)][??k'=k][\CH(a \rightarrow k')][c_q!k']([c_c!d]K_A + [c_d!d][c_c!k]K_A') \\
    B &:= [c_b?b][q?\ell][\CH(b \rightarrow \ell)]([c_c?d]K_B + [c_d?d][c_c?k][??v=k][\CX(\ell \rightarrow v)]K_B') \\
  \end{aligned}
\end{gather*}

These processes are meant to be run as 
$(A \parcomp B) \setminus \{c_b, c_c, c_d\}$.
We will omit the $\setminus \{c_b, c_c, c_d\}$ part for brevity, 
while keeping in mind that communication on those channels do not 
propagate out.

The advantage of spatial compositionality is that we can model Alice 
and Bob as an open system that does not explicitly include the 
attacker.
The shape of the most interesting Alice-Bob evolution to analyze is, 
of course, the one where qubit $k'$ leaves the duo temporarily:
\begin{equation}\label{eq:bb84-reductions}
  \st{A\parcomp B}{\rho_0(\FV (A\parcomp B))} \Sstep{c_q!q_{k'}} 
  \st{A_1 \parcomp B_1}{\rho_1} \Sstep{c_q?q_\ell} 
  \st{A_2 \parcomp B_2}{\rho_2} \Sstep{\tau} \st{K_A \parcomp K_B}{\rho}
\end{equation}
where $\dom q_\ell = \{\ell\}$ and $q_{k'} = \{k'\}$.
(Showing the case where Alice and Bob decide to accept.)
The qubit $q_{k'}$ is intercepted by an external Eve, who may return 
the qubit without touching it, perturb and return it, or return an 
entirely unrelated qubit.
Either way, the configuration $\st{A_1 \parcomp B_1}{\rho_1}$ is one 
that cannot be adequately modeled by density matrices because, when 
the returned qubit $q_\ell$ does turn out to be the same one that 
left, the configuration must be restored to a pure one in order to 
keep reasoning about the evolution thereafter.

And we can indeed reason about the remaining evolution!
In fact, if we focus on the happy case where the attacker only 
modifies $q_{k'}$ in a way that leaves no trace of external qubits in 
its descriptor, the Security clause above is rather simple to prove.


\begin{proposition}\label{thm:bb84-simple}
  In \cref{eq:bb84-reductions}, let $\st{A_1 \parcomp B_1}{\rho_1}$ 
  be the last configuration before the $c_q?q_\ell$ action and assume 
  $\rho_1 \cup q_\ell$ is valid.
  If $\ell = k'$ and 
  $\supp q_\ell \subseteq \dom \rho_1 \cup \{k'\}$, then for any 
  valid full extension $\mu \supseteq \rho$, no qubit 
  $g \in \dom \mu - \dom \rho$ can be correlated with $k$ when 
  measured in the Z axis.
\end{proposition}

In other words, if the same $k'$ came back in $q_\ell$, Eve can't 
possess a correlated qubit.
The extension $\mu$, being arbitrary, can contain everything Eve has; 
yet, we can't find a correlated $g$ in it, except in $\rho$ which 
Alice and Bob own (recall they don't discard, so $\rho$ contains 
every qubit they touched).
The requirement 
$\ell = k' \land \supp q_\ell \subseteq \dom \rho_1 \cup \{k'\}$ is 
fairly strong, but not trivial: for example, Eve could have performed 
any single-qubit unitary on $k'$.
What she could not have done is a CX or anything that might entangle 
it with her other qubits.

To prove the proposition, we note that if Eve's qubits didn't spill 
into $\ell$'s descriptors, then Alice and Bob's qubits form a closed 
system unentangled with the outside.
To formalize this argument, we will gather the mutual reference 
structure of qubits in a graph.

\begin{definition}
  Given a store $\rho$, define its \emph{mention graph} to have 
  $\supp \rho$ as the set of nodes and have edges $x \rightarrow y$ 
  iff $x \neq y$ and $y \in \supp \rho(x)$.
\end{definition}

\begin{lemma}\label{thm:mention-graph-no-sink}
  The mention graph of a valid full store $\rho$ has no sinks, i.e.\ 
  nodes where some edges come in but none go out.
  \begin{proof}
    Let $y$ have an edge to a sink $x$.
    Then $\rho|_{\{x\}}$ is itself a valid full store, so by 
    \cref{thm:completion}, 
    $\rho|_{\{x\}} = U^\dagger \rho_0(x) U$ where $U$ acts 
    trivially on every qubit except $x$, i.e.\ $U = \iota_{\{x\}}U'$ 
    for some unitary $U'$.
    Each $\rho(y)_j$ ($j \in \{X,Z\}$) commutes with $\rho(x)_X$ and 
    $\rho(x)_Z$, so $A := U \rho(y)_j U^\dagger$ commutes with $X_x$ 
    and $Z_x$.
    We can expand out 
    $A = a\otimes 1 + b\otimes X_x + c\otimes Z_x + d\otimes X_z Z_x$ 
    where $a,b,c,d$ can be written in those $X_z$ and $Z_z$ with 
    $z \neq x$.
    Commutation with $X_x$, $Z_x$ forces $b = c = d = 0$, so 
    $A = a \otimes 1$, which means 
    $\rho(y)_j = U^\dagger (a \otimes 1) U = a \otimes ({U'}^\dagger 
    1 U') = a\otimes 1$; therefore $x \notin \supp \rho(y)_j$.
  \end{proof}
\end{lemma}

\begin{proof}[Proof of \cref{thm:bb84-simple}]
  Let $\ket{\mu}$ be the state encoded by $\mu$.
  Recall that $Z$-measurements on $g$, $k$ are uncorrelated on $\mu$ 
  iff 
  $\bracket{\mu}{Z_gZ_k}{\mu} = \bracket{\mu}{Z_g}{\mu} 
  \bracket{\mu}{Z_k}{\mu} = 
  \bracket{0^W}{\mu(g)_Z}{0^W}\bracket{0^W}{\mu(x)_Z}{0^W}$.
  For this, it suffices to show $\mu(g)_Z$ acts trivially on $k$ and 
  vice versa.
  The assumption $q_\ell \subseteq \dom \rho_1\cup \{k'\}$ ensures 
  that merging $q_\ell$ into $\rho_1$ does not introduce any 
  variables other than $k'$.
  Because it is the last store before $q_\ell$ is incorporated, 
  $\rho_1 \cup q_\ell$ is the store right after receiving 
  $c_q?q_\ell$.
  But $\supp \rho_1 - \dom \rho_1 = \{k'\}$, because $k'$ is the only 
  qubit that left the Alice-Bob system up to that point.
  Therefore, Alice and Bob's joint store becomes full after that 
  receive operation, and fullness is preserved by subsequent 
  evolution, all the way up to $\rho$.
  This means that, by \cref{thm:mention-graph-no-sink}, the 
  descriptors in $\mu - \rho$ cannot reference variables in 
  $\dom \rho$.
  Hence $\mu(g)_Z$ acts trivially on $k$, and $\mu(k)_Z$ acts 
  trivially on $g$.
\end{proof}

The argumentation above demonstrates a reasoning pattern unique to 
DH-CCS: it can tell when absence of entanglement is \emph{restored} 
by external communication.

BB84 also exposes a current limitation with reasoning in DH-CCS.
If we try to reason about the case where $q_\ell$ does mention 
external qubits, we hit against the fact that qubits can mention 
others without being entangled.
For example, $\CX(x \rightarrow y)$ takes $\rho_0(x,y)$ to 
$[x \mapsto (X_x X_y,Z_x), y \mapsto (X_y, Z_x Z_y)]$, but CX on the 
initial state $\ket{00}$ should be a no-op.
Instead, $x$ and $y$ mention each other.

The problem here is that DH descriptors capture the evolution of the 
$\Obs$ part of $\langle 0^V|\Obs|0^V \rangle$ detached from the 
$\langle 0^V| \_ | 0^V \rangle$ part.
The map $q \mapsto U_{\CX}^\dagger q U_\CX$ is not a no-op, because 
what if the result was applied to something other than $|0^V\rangle$?
Just in case the tracked observable is used with a different initial 
state, DH descriptors record any operations that \emph{might} create 
entanglement as mutual mentions.
Thus the mentions graph is an over-approximation of entanglement 
structure.
The Bell state example from \cref{sec:bisimulation} is similar: 
$H(x).\CX(x \rightarrow y)$ and $H(y).\CX(y \rightarrow x)$ map 
$\rho_0(x,y)$ differently because these operations would diverge if 
the initial state were $\ket{0+}$.

Put differently, local descriptor inspection can establish absence of 
information leakage in open systems, but not presence.
This is a limitation coming from our lack of understanding about how 
to properly collapse the gauge freedom of DH descriptors.
It seems accounting for the initial state $|0^V\rangle$ should give a 
significant collapse, but what kind of equalities that induces on 
descriptors, or whether that takes care of all unwanted degrees of 
freedom, is currently unknown.

For now, the need for accurate accounting of information flow must be 
met by collapsing DH descriptors to density matrices.
The proof strategy for full BB84 would be to convert stores to 
density matrices and apply the standard security argument like we 
sketched at the beginning of this section.
The utility of DH descriptors is then to track the movements of 
qubits across process and system boundaries.
Information flow is monitored by current-density arguments, just like 
with program equivalence.

\section{Related Works}
\label{sec:related-works}

Deutsch-Hayden descriptors were introduced by the namesake authors in 
\cite{DeutschHayden2000InformationFlowEntangled}.
Their purpose was to show quantum mechanics to be \emph{local}, that 
quantum information is localized on each qubit and operations on one 
qubit has no effect on another until the qubits are made to interact.
Horsman and Vedral \cite{HorsmanVedral2007DevelopingDeutschHayden} 
developed this further to handle relative states and partial 
descriptor sets like we use here.
Raymond-Robichaud 
\cite{Raymond-Robichaud2017LequivalenceEntreLocalrealisme} has 
proposed alternative local theories, but these were later shown to be 
equivalent to DH \cite{Bedard2021CostQuantumLocality}.

There is some debate on whether DH actually demonstrated locality.
Timpson \cite{Timpson2005NonlocalityInformationFlow} counters, e.g.\ 
that the initial state $\ket{0^n}$ implicit in an $n$-qubit DH 
descriptor serves as a global connection.
Wallace and Timpson \cite{WallaceTimpson2007NonlocalityGaugeFreedom} 
go further and claim that DH's gauge freedom shows that locality is a 
mirage, living only in the unphysical degrees of freedom.
They base this assertion on the belief that gauge freedom should be 
collapsed until the resulting description coincides with density 
matrices, since the latter encodes all measurable quantities.
But they only consider what amounts to current-density observation.
Viewing observation as a continuous monitoring process makes part of 
this freedom semantically significant (\cref{sec:bisimulation}).

In our setting, the most important feature of DH descriptors is that 
they provide modular accounting of quantum information, and whether 
nature works that way is a moot point (as interesting a question as 
that may be).
As noted by B\'{e}dard \cite{Bedard2021ABCDeutschHayden}, it appears 
to be still widely believed that entanglement makes quantum states 
inherently non-modular---a myth that B\'{e}dard sought to dispel.
This paper doubles down on this idea by adopting a modular state 
representation and demonstrating a kind of compositionality hitherto 
unavailable to quantum process calculi.

A number of quantum process calculi have been proposed over the 
years.
The earliest efforts such as QPAlg 
\cite{LalireJorrand2004ProcessAlgebraicApproach} and CQP 
\cite{GayNagarajan2005CommunicatingQuantumProcesses} strove to unify 
quantum and classical information processing in one framework.
Later calculi such as qCCS 
\cite{FengEtAl2007ProbabilisticBisimulationsQuantum} extended the 
modeled operations while simplifying and strengthening the mechanisms 
that ensure physical realizability.
Our calculus's design is greatly indebted to qCCS, along with lqCCS 
\cite{CeragioliEtAl2024QuantumBisimilarityBarbs}.
A prime novelty of qCCS was the handling of ``input and output of 
quantum systems which are possibly [entangled]'' 
\cite{FengEtAl2007ProbabilisticBisimulationsQuantum}, but this only 
applied to communication between in-model processes or with 
unentangled external processes.
We expand on this capability by adding communication with external 
entangled entities.

Despite substantial differences in syntax, communication discipline, 
typing discipline, and equivalence notions, one feature is remarkably 
consistent throughout these calculi, which is the monolithic 
representation of quantum state.
All qubits are put in a single, global state vector or density 
matrix, and configurations are those along with process terms; then a 
transition system is defined on ensembles of those.
Communication transfers only qubit names, which cannot be interpreted 
independently of the global quantum state.
(An exception is qCCS's external I/O, which requires unentangled 
input sources and traces out output qubits.)
DH-CCS instead employs a state representation where individual qubits 
can be pushed around as standalone entities, which leads to spatial 
compositionality and open-system modeling.

Beyond the original calculi themselves, considerable effort has gone 
into defining behavioral equivalences, including observational 
equivalence \cite{YasudaEtAl2014ObservationalEquivalenceUsing}, open 
and ground bisimulations 
\cite{QinEtAl2020VerifyingQuantumCommunication}, branching 
bisimulation \cite{WuEtAl2024BranchingBisimulationSemantics}, and 
effect-based semantics 
\cite{CeragioliEtAl2024EffectSemanticsQuantum}.
A scathing critique came from Ceragioli et al., who pointed out that 
``the soundness of [the observable properties proposed in earlier 
calculi] has never been validated against the prescriptions of 
quantum theory'' \cite{CeragioliEtAl2024QuantumBisimilarityBarbs}.
They show a specific pathology with ensemble semantics and propose a 
tightened notion of process equivalence.

Partly in a nod to this critique, we have thoroughly grounded our 
derivation of observer equivalence on a detailed rationale about what 
observations should be physically possible, and what might perturb 
it.
We of course do not claim to have the final word on this topic, but 
we hope to at least contribute to a thorougher understanding of 
quantum process equivalence by providing a design that is easy to 
scrutinize and perhaps generalizes to other calculi.

\section{Methods}

Interactions with MS Copilot and ChatGPT influenced the formulation 
of $\Domain$, spatial compositionality, and bisimulation.
Copilot generated the proofs of 
\cref{thm:completion,thm:mention-graph-no-sink}, as well as a program 
to check DH descriptor calculations.
Both were used to survey related works.
We have verified all AI-generated materials and assume full 
responsibility for their correctness and relevance.

\section{Conclusion}
\label{sec:conclusion}

This paper established DH-CCS, a quantum process calculus based on 
the Deutsch-Hayden descriptor formulation of quantum mechanics.
Its distinguishing feature is a quantum state representation that can 
be split and merged along qubit boundaries regardless of 
entanglement.
The resulting calculus is spatially compositional, allowing 
individual subprocesses' execution to be tracked independently, in a 
way that can be put back together to recover the global execution 
result.
The upshot is a calculus capable of modeling fully open systems, 
where qubits can enter, leave, and come back, all without losing 
entanglement information.
It explicitly unifies allocation, discarding, measurement (though not 
classical branching) as communication with an external entity.

We have defined a notion of process equivalence with the explicit 
goal of taming the gauge (extra) freedom of DH descriptors, equipped 
with an approximating bisimulation.
The BB84 study suggests that DH descriptors are well-suited to 
tracking movement of quantum data, whereas reasoning about 
information leakage still requires collapsing to density matrices.

There are two particularly interesting directions for future work.
One is whether classical control can be incorporated without 
introducing ensembles (i.e.\ classical probability distributions).
The most difficult question there would be what to do about ownership 
predicated on the value of a qubit.
Perhaps a semantics that allows the process term to be in quantum 
superposition may work.

The other interesting question is whether we can characterize and 
collapse precisely the unobservable gauge freedom of DH descriptors.
We can't help but feel that we've seen this movie before in 
programming language theory.
When Plotkin discovered the mismatch between observable behaviors of 
sequential programs and Scott denotations 
\cite{Plotkin1977LCFConsideredProgramming}, the gap was first bridged 
by a syntactic model \cite{Milner1977FullyAbstractModels}, then 
eventually game semantics 
\cite{AbramskyEtAl2000FullAbstractionPCF,HylandOng2000FullAbstractionPCF}.
The difference is that gauge freedom of DH descriptors have been 
well-known for a while, but a characterization of the gap like 
parallel-or does not appear to be forthcoming; so, there is no 
Plotkin in this movie.

We could not find in the DH descriptor literature a good 
characterization separating the degrees of freedom that matter from 
those that do not.
Can our continuous monitoring perspective serve as such a 
characterization?
Our historical semantics is a kind of syntactic model, but it's 
infinitely redundant.
Can we quotient it while accounting for the distributed nature of DH 
descriptors?
Can something like game semantics provide a syntax-free model that 
collapses the freedom exactly?
If the characterization and collapse can be successfully made, that 
could shed a new light on the concept of locality and the nature of 
quantum information flow in a distributed system.


\bibliographystyle{ACM-Reference-Format}
\bibliography{local}


\begin{thebibliography}{24}


\ifx \showCODEN    \undefined \def \showCODEN     #1{\unskip}     \fi
\ifx \showISBNx    \undefined \def \showISBNx     #1{\unskip}     \fi
\ifx \showISBNxiii \undefined \def \showISBNxiii  #1{\unskip}     \fi
\ifx \showISSN     \undefined \def \showISSN      #1{\unskip}     \fi
\ifx \showLCCN     \undefined \def \showLCCN      #1{\unskip}     \fi
\ifx \shownote     \undefined \def \shownote      #1{#1}          \fi
\ifx \showarticletitle \undefined \def \showarticletitle #1{#1}   \fi
\ifx \showURL      \undefined \def \showURL       {\relax}        \fi
\providecommand\bibfield[2]{#2}
\providecommand\bibinfo[2]{#2}
\providecommand\natexlab[1]{#1}
\providecommand\showeprint[2][]{arXiv:#2}

\bibitem[Abramsky et~al\mbox{.}(2000)]%
        {AbramskyEtAl2000FullAbstractionPCF}
\bibfield{author}{\bibinfo{person}{Samson Abramsky}, \bibinfo{person}{Radha
  Jagadeesan}, {and} \bibinfo{person}{Pasquale Malacaria}.}
  \bibinfo{year}{2000}\natexlab{}.
\newblock \showarticletitle{Full {{Abstraction}} for {{PCF}}}.
\newblock \bibinfo{journal}{\emph{Information and Computation}}
  \bibinfo{volume}{163}, \bibinfo{number}{2} (\bibinfo{date}{Dec.}
  \bibinfo{year}{2000}), \bibinfo{pages}{409--470}.
\newblock
\showISSN{0890-5401}
\href{https://doi.org/10.1006/inco.2000.2930}{doi:\nolinkurl{10.1006/inco.2000.2930}}


\bibitem[B{\'e}dard(2021a)]%
        {Bedard2021ABCDeutschHayden}
\bibfield{author}{\bibinfo{person}{Charles~Alexandre B{\'e}dard}.}
  \bibinfo{year}{2021}\natexlab{a}.
\newblock \showarticletitle{The {{ABC}} of {{Deutsch}}--{{Hayden
  Descriptors}}}.
\newblock \bibinfo{journal}{\emph{Quantum Reports}} \bibinfo{volume}{3},
  \bibinfo{number}{2} (\bibinfo{date}{June} \bibinfo{year}{2021}),
  \bibinfo{pages}{272--285}.
\newblock
\showISSN{2624-960X}
\href{https://doi.org/10.3390/quantum3020017}{doi:\nolinkurl{10.3390/quantum3020017}}


\bibitem[B{\'e}dard(2021b)]%
        {Bedard2021CostQuantumLocality}
\bibfield{author}{\bibinfo{person}{C.~A. B{\'e}dard}.}
  \bibinfo{year}{2021}\natexlab{b}.
\newblock \showarticletitle{The Cost of Quantum Locality}.
\newblock \bibinfo{journal}{\emph{Proceedings of the Royal Society A:
  Mathematical, Physical and Engineering Sciences}} \bibinfo{volume}{477},
  \bibinfo{number}{2246} (\bibinfo{date}{Feb.} \bibinfo{year}{2021}),
  \bibinfo{pages}{20200602}.
\newblock
\showISSN{1364-5021}
\href{https://doi.org/10.1098/rspa.2020.0602}{doi:\nolinkurl{10.1098/rspa.2020.0602}}


\bibitem[Bennett and Brassard(2014)]%
        {BennettBrassard2014QuantumCryptographyPublic}
\bibfield{author}{\bibinfo{person}{Charles~H. Bennett} {and}
  \bibinfo{person}{Gilles Brassard}.} \bibinfo{year}{2014}\natexlab{}.
\newblock \showarticletitle{Quantum Cryptography: {{Public}} Key Distribution
  and Coin Tossing}.
\newblock \bibinfo{journal}{\emph{Theoretical Computer Science}}
  \bibinfo{volume}{560} (\bibinfo{date}{Dec.} \bibinfo{year}{2014}),
  \bibinfo{pages}{7--11}.
\newblock
\showISSN{03043975}
\showeprint[arxiv]{2003.06557}~[quant-ph]
\href{https://doi.org/10.1016/j.tcs.2014.05.025}{doi:\nolinkurl{10.1016/j.tcs.2014.05.025}}


\bibitem[Ceragioli et~al\mbox{.}(2024a)]%
        {CeragioliEtAl2024EffectSemanticsQuantum}
\bibfield{author}{\bibinfo{person}{Lorenzo Ceragioli}, \bibinfo{person}{Fabio
  Gadducci}, \bibinfo{person}{Giuseppe Lomurno}, {and}
  \bibinfo{person}{Gabriele Tedeschi}.} \bibinfo{year}{2024}\natexlab{a}.
\newblock \showarticletitle{Effect {{Semantics}} for {{Quantum Process
  Calculi}}}. In \bibinfo{booktitle}{\emph{35th {{International Conference}} on
  {{Concurrency Theory}} ({{CONCUR}} 2024)}} \emph{(\bibinfo{series}{Leibniz
  {{International Proceedings}} in {{Informatics}} ({{LIPIcs}})},
  Vol.~\bibinfo{volume}{311})}, \bibfield{editor}{\bibinfo{person}{Rupak
  Majumdar} {and} \bibinfo{person}{Alexandra Silva}} (Eds.).
  \bibinfo{publisher}{Schloss Dagstuhl -- Leibniz-Zentrum f\"ur Informatik},
  \bibinfo{address}{Dagstuhl, Germany}, \bibinfo{pages}{16:1--16:22}.
\newblock
\showISBNx{978-3-95977-339-3}
\showISSN{1868-8969}
\href{https://doi.org/10.4230/LIPIcs.CONCUR.2024.16}{doi:\nolinkurl{10.4230/LIPIcs.CONCUR.2024.16}}


\bibitem[Ceragioli et~al\mbox{.}(2024b)]%
        {CeragioliEtAl2024QuantumBisimilarityBarbs}
\bibfield{author}{\bibinfo{person}{Lorenzo Ceragioli}, \bibinfo{person}{Fabio
  Gadducci}, \bibinfo{person}{Giuseppe Lomurno}, {and}
  \bibinfo{person}{Gabriele Tedeschi}.} \bibinfo{year}{2024}\natexlab{b}.
\newblock \showarticletitle{Quantum {{Bisimilarity}} via {{Barbs}} and
  {{Contexts}}: {{Curbing}} the {{Power}} of {{Non-deterministic Observers}}}.
\newblock \bibinfo{journal}{\emph{Proc. ACM Program. Lang.}}
  \bibinfo{volume}{8}, \bibinfo{number}{POPL} (\bibinfo{date}{Jan.}
  \bibinfo{year}{2024}), \bibinfo{pages}{43:1269--43:1297}.
\newblock
\href{https://doi.org/10.1145/3632885}{doi:\nolinkurl{10.1145/3632885}}


\bibitem[Deutsch and Hayden(2000)]%
        {DeutschHayden2000InformationFlowEntangled}
\bibfield{author}{\bibinfo{person}{David Deutsch} {and}
  \bibinfo{person}{Patrick Hayden}.} \bibinfo{year}{2000}\natexlab{}.
\newblock \showarticletitle{Information Flow in Entangled Quantum Systems}.
\newblock \bibinfo{journal}{\emph{Proceedings of the Royal Society A:
  Mathematical, Physical and Engineering Sciences}} \bibinfo{volume}{456},
  \bibinfo{number}{1999} (\bibinfo{date}{July} \bibinfo{year}{2000}),
  \bibinfo{pages}{1759--1774}.
\newblock
\showISSN{1364-5021}
\href{https://doi.org/10.1098/rspa.2000.0585}{doi:\nolinkurl{10.1098/rspa.2000.0585}}


\bibitem[Everett(1957)]%
        {Everett1957RelativeStateFormulation}
\bibfield{author}{\bibinfo{person}{Hugh Everett}.}
  \bibinfo{year}{1957}\natexlab{}.
\newblock \showarticletitle{"{{Relative State}}" {{Formulation}} of {{Quantum
  Mechanics}}}.
\newblock \bibinfo{journal}{\emph{Reviews of Modern Physics}}
  \bibinfo{volume}{29}, \bibinfo{number}{3} (\bibinfo{year}{1957}),
  \bibinfo{pages}{454--462}.
\newblock
\href{https://doi.org/10.1103/RevModPhys.29.454}{doi:\nolinkurl{10.1103/RevModPhys.29.454}}


\bibitem[Feng et~al\mbox{.}(2007)]%
        {FengEtAl2007ProbabilisticBisimulationsQuantum}
\bibfield{author}{\bibinfo{person}{Yuan Feng}, \bibinfo{person}{Runyao Duan},
  \bibinfo{person}{Zhengfeng Ji}, {and} \bibinfo{person}{Mingsheng Ying}.}
  \bibinfo{year}{2007}\natexlab{}.
\newblock \showarticletitle{Probabilistic Bisimulations for Quantum Processes}.
\newblock \bibinfo{journal}{\emph{Information and Computation}}
  \bibinfo{volume}{205}, \bibinfo{number}{11} (\bibinfo{date}{Nov.}
  \bibinfo{year}{2007}), \bibinfo{pages}{1608--1639}.
\newblock
\showISSN{0890-5401}
\href{https://doi.org/10.1016/j.ic.2007.08.001}{doi:\nolinkurl{10.1016/j.ic.2007.08.001}}


\bibitem[Gabbay and Pitts(2002)]%
        {GabbayPitts2002NewApproachAbstract}
\bibfield{author}{\bibinfo{person}{Murdoch~J. Gabbay} {and}
  \bibinfo{person}{Andrew~M. Pitts}.} \bibinfo{year}{2002}\natexlab{}.
\newblock \showarticletitle{A {{New Approach}} to {{Abstract Syntax}} with
  {{Variable Binding}}}.
\newblock \bibinfo{journal}{\emph{Formal Aspects of Computing}}
  \bibinfo{volume}{13}, \bibinfo{number}{3-5} (\bibinfo{date}{July}
  \bibinfo{year}{2002}), \bibinfo{pages}{341--363}.
\newblock
\showISSN{0934-5043}
\href{https://doi.org/10.1007/s001650200016}{doi:\nolinkurl{10.1007/s001650200016}}


\bibitem[Gay and Nagarajan(2005)]%
        {GayNagarajan2005CommunicatingQuantumProcesses}
\bibfield{author}{\bibinfo{person}{Simon~J. Gay} {and}
  \bibinfo{person}{Rajagopal Nagarajan}.} \bibinfo{year}{2005}\natexlab{}.
\newblock \showarticletitle{Communicating Quantum Processes}. In
  \bibinfo{booktitle}{\emph{Proceedings of the 32nd {{ACM SIGPLAN-SIGACT}}
  Symposium on {{Principles}} of Programming Languages}}
  \emph{(\bibinfo{series}{{{POPL}} '05})}. \bibinfo{publisher}{Association for
  Computing Machinery}, \bibinfo{address}{New York, NY, USA},
  \bibinfo{pages}{145--157}.
\newblock
\showISBNx{978-1-58113-830-6}
\href{https://doi.org/10.1145/1040305.1040318}{doi:\nolinkurl{10.1145/1040305.1040318}}


\bibitem[Horsman and Vedral(2007)]%
        {HorsmanVedral2007DevelopingDeutschHayden}
\bibfield{author}{\bibinfo{person}{Dominic Horsman} {and} \bibinfo{person}{V
  Vedral}.} \bibinfo{year}{2007}\natexlab{}.
\newblock \showarticletitle{Developing the {{Deutsch}}--{{Hayden}} Approach to
  Quantum Mechanics}.
\newblock \bibinfo{journal}{\emph{New Journal of Physics}} \bibinfo{volume}{9},
  \bibinfo{number}{5} (\bibinfo{date}{May} \bibinfo{year}{2007}),
  \bibinfo{pages}{135}.
\newblock
\showISSN{1367-2630}
\href{https://doi.org/10.1088/1367-2630/9/5/135}{doi:\nolinkurl{10.1088/1367-2630/9/5/135}}


\bibitem[Hyland and Ong(2000)]%
        {HylandOng2000FullAbstractionPCF}
\bibfield{author}{\bibinfo{person}{J.~M.~E. Hyland} {and}
  \bibinfo{person}{C.~H.~L. Ong}.} \bibinfo{year}{2000}\natexlab{}.
\newblock \showarticletitle{On {{Full Abstraction}} for {{PCF}}: {{I}}, {{II}},
  and {{III}}}.
\newblock \bibinfo{journal}{\emph{Information and Computation}}
  \bibinfo{volume}{163}, \bibinfo{number}{2} (\bibinfo{date}{Dec.}
  \bibinfo{year}{2000}), \bibinfo{pages}{285--408}.
\newblock
\showISSN{0890-5401}
\href{https://doi.org/10.1006/inco.2000.2917}{doi:\nolinkurl{10.1006/inco.2000.2917}}


\bibitem[Jacobson(2009)]%
        {Jacobson2009BasicAlgebraII}
\bibfield{author}{\bibinfo{person}{Nathan Jacobson}.}
  \bibinfo{year}{2009}\natexlab{}.
\newblock \bibinfo{booktitle}{\emph{Basic {{Algebra II}}: {{Second Edition}}}}.
\newblock \bibinfo{publisher}{Dover Publications}, \bibinfo{address}{Mineola,
  NY}.
\newblock
\showISBNx{978-0-486-47187-7}


\bibitem[Lalire and Jorrand(2004)]%
        {LalireJorrand2004ProcessAlgebraicApproach}
\bibfield{author}{\bibinfo{person}{Marie Lalire} {and}
  \bibinfo{person}{Philippe Jorrand}.} \bibinfo{year}{2004}\natexlab{}.
\newblock \bibinfo{title}{A {{Process Algebraic Approach}} to {{Concurrent}}
  and {{Distributed Quantum Computation}}: {{Operational Semantics}}}.
\newblock
\showeprint[arxiv]{quant-ph/0407005}
\href{https://doi.org/10.48550/arXiv.quant-ph/0407005}{doi:\nolinkurl{10.48550/arXiv.quant-ph/0407005}}


\bibitem[Milner(1977)]%
        {Milner1977FullyAbstractModels}
\bibfield{author}{\bibinfo{person}{Robin Milner}.}
  \bibinfo{year}{1977}\natexlab{}.
\newblock \showarticletitle{Fully Abstract Models of Typed
  {$\lambda$}-Calculi}.
\newblock \bibinfo{journal}{\emph{Theoretical Computer Science}}
  \bibinfo{volume}{4}, \bibinfo{number}{1} (\bibinfo{date}{Feb.}
  \bibinfo{year}{1977}), \bibinfo{pages}{1--22}.
\newblock
\showISSN{0304-3975}
\href{https://doi.org/10.1016/0304-3975(77)90053-6}{doi:\nolinkurl{10.1016/0304-3975(77)90053-6}}


\bibitem[Milner and Sangiorgi(1992)]%
        {MilnerSangiorgi1992BarbedBisimulation}
\bibfield{author}{\bibinfo{person}{Robin Milner} {and} \bibinfo{person}{Davide
  Sangiorgi}.} \bibinfo{year}{1992}\natexlab{}.
\newblock \showarticletitle{Barbed Bisimulation}. In
  \bibinfo{booktitle}{\emph{Automata, {{Languages}} and {{Programming}}}},
  \bibfield{editor}{\bibinfo{person}{W.~Kuich}} (Ed.).
  \bibinfo{publisher}{Springer}, \bibinfo{address}{Berlin, Heidelberg},
  \bibinfo{pages}{685--695}.
\newblock
\showISBNx{978-3-540-47278-0}
\href{https://doi.org/10.1007/3-540-55719-9_114}{doi:\nolinkurl{10.1007/3-540-55719-9_114}}


\bibitem[Plotkin(1977)]%
        {Plotkin1977LCFConsideredProgramming}
\bibfield{author}{\bibinfo{person}{G.~D. Plotkin}.}
  \bibinfo{year}{1977}\natexlab{}.
\newblock \showarticletitle{{{LCF}} Considered as a Programming Language}.
\newblock \bibinfo{journal}{\emph{Theoretical Computer Science}}
  \bibinfo{volume}{5}, \bibinfo{number}{3} (\bibinfo{date}{Dec.}
  \bibinfo{year}{1977}), \bibinfo{pages}{223--255}.
\newblock
\showISSN{0304-3975}
\href{https://doi.org/10.1016/0304-3975(77)90044-5}{doi:\nolinkurl{10.1016/0304-3975(77)90044-5}}


\bibitem[Qin et~al\mbox{.}(2020)]%
        {QinEtAl2020VerifyingQuantumCommunication}
\bibfield{author}{\bibinfo{person}{Xudong Qin}, \bibinfo{person}{Yuxin Deng},
  {and} \bibinfo{person}{Wenjie Du}.} \bibinfo{year}{2020}\natexlab{}.
\newblock \showarticletitle{Verifying {{Quantum Communication Protocols}} with
  {{Ground Bisimulation}}}.
\newblock \bibinfo{journal}{\emph{Tools and Algorithms for the Construction and
  Analysis of Systems}}  \bibinfo{volume}{12079} (\bibinfo{date}{March}
  \bibinfo{year}{2020}), \bibinfo{pages}{21--38}.
\newblock
\href{https://doi.org/10.1007/978-3-030-45237-7_2}{doi:\nolinkurl{10.1007/978-3-030-45237-7_2}}


\bibitem[{Raymond-Robichaud}(2017)]%
        {Raymond-Robichaud2017LequivalenceEntreLocalrealisme}
\bibfield{author}{\bibinfo{person}{Paul {Raymond-Robichaud}}.}
  \bibinfo{year}{2017}\natexlab{}.
\newblock \showarticletitle{L'\'equivalence Entre Le Local-R\'ealisme et Le
  Principe de Non-Signalement}.
\newblock  (\bibinfo{date}{Aug.} \bibinfo{year}{2017}).
\newblock
\showeprint[hdl]{1866/20497}
\href{https://doi.org/10.71781/10506}{doi:\nolinkurl{10.71781/10506}}


\bibitem[Timpson(2005)]%
        {Timpson2005NonlocalityInformationFlow}
\bibfield{author}{\bibinfo{person}{C.~G. Timpson}.}
  \bibinfo{year}{2005}\natexlab{}.
\newblock \showarticletitle{Nonlocality and {{Information Flow}}: {{The
  Approach}} of {{Deutsch}} and {{Hayden}}}.
\newblock \bibinfo{journal}{\emph{Foundations of Physics}}
  \bibinfo{volume}{35}, \bibinfo{number}{2} (\bibinfo{date}{Feb.}
  \bibinfo{year}{2005}), \bibinfo{pages}{313--343}.
\newblock
\showISSN{1572-9516}
\href{https://doi.org/10.1007/s10701-004-1946-1}{doi:\nolinkurl{10.1007/s10701-004-1946-1}}


\bibitem[Wallace and Timpson(2007)]%
        {WallaceTimpson2007NonlocalityGaugeFreedom}
\bibfield{author}{\bibinfo{person}{David Wallace} {and}
  \bibinfo{person}{Christopher~G. Timpson}.} \bibinfo{year}{2007}\natexlab{}.
\newblock \showarticletitle{Non-Locality and {{Gauge Freedom}} in {{Deutsch}}
  and {{Hayden}}'s {{Formulation}} of {{Quantum Mechanics}}}.
\newblock \bibinfo{journal}{\emph{Foundations of Physics}}
  \bibinfo{volume}{37}, \bibinfo{number}{6} (\bibinfo{date}{June}
  \bibinfo{year}{2007}), \bibinfo{pages}{951--955}.
\newblock
\showISSN{1572-9516}
\href{https://doi.org/10.1007/s10701-007-9135-7}{doi:\nolinkurl{10.1007/s10701-007-9135-7}}


\bibitem[Wu et~al\mbox{.}(2024)]%
        {WuEtAl2024BranchingBisimulationSemantics}
\bibfield{author}{\bibinfo{person}{Hao Wu}, \bibinfo{person}{Qizhe Yang}, {and}
  \bibinfo{person}{Huan Long}.} \bibinfo{year}{2024}\natexlab{}.
\newblock \showarticletitle{Branching Bisimulation Semantics for Quantum
  Processes}.
\newblock \bibinfo{journal}{\emph{Inform. Process. Lett.}}
  \bibinfo{volume}{186} (\bibinfo{date}{Aug.} \bibinfo{year}{2024}),
  \bibinfo{pages}{106492}.
\newblock
\showISSN{0020-0190}
\href{https://doi.org/10.1016/j.ipl.2024.106492}{doi:\nolinkurl{10.1016/j.ipl.2024.106492}}


\bibitem[Yasuda et~al\mbox{.}(2014)]%
        {YasudaEtAl2014ObservationalEquivalenceUsing}
\bibfield{author}{\bibinfo{person}{Kazuya Yasuda}, \bibinfo{person}{Takahiro
  Kubota}, {and} \bibinfo{person}{Yoshihiko Kakutani}.}
  \bibinfo{year}{2014}\natexlab{}.
\newblock \showarticletitle{Observational {{Equivalence Using Schedulers}} for
  {{Quantum Processes}}}.
\newblock \bibinfo{journal}{\emph{Electronic Proceedings in Theoretical
  Computer Science}}  \bibinfo{volume}{172} (\bibinfo{date}{Dec.}
  \bibinfo{year}{2014}), \bibinfo{pages}{191--203}.
\newblock
\showISSN{2075-2180}
\showeprint[arxiv]{1412.8546}~[cs.LO]
\href{https://doi.org/10.4204/EPTCS.172.13}{doi:\nolinkurl{10.4204/EPTCS.172.13}}


\end{thebibliography}
\end{document}